\documentclass[lettersize,journal]{IEEEtran}
\usepackage{cite}
\usepackage{amsmath,amsthm,amssymb,,amsfonts}
\usepackage{algorithmic}
\usepackage{textcomp}

\usepackage[ruled,vlined]{algorithm2e}
\usepackage{dsfont}
\usepackage{mathrsfs}
\usepackage{mathbbol}

\usepackage{accents}

\usepackage{float}

\usepackage{xcolor,graphicx,epstopdf}
\usepackage{dsfont}
\usepackage{hyperref}
\usepackage{caption}
\usepackage{subcaption}

\usepackage[subfigure]{tocloft}

\usepackage{gensymb}

\usepackage{textcomp,mathcomp}
\newtheorem{theorem}{Theorem}
\newtheorem{definition}{Definition}

\newtheorem{lemma}{Lemma}

\newtheorem{proposition}{Proposition}
\newtheorem{assumption}{Assumption}

\begin{document}
\title{Competitive Equilibrium for Dynamic Multi-Agent Systems: Social Shaping and Price Trajectories}
\author{Zeinab Salehi, Yijun Chen, Elizabeth L. Ratnam, Ian R. Petersen, and Guodong Shi
\thanks{Z. Salehi, E. L. Ratnam, and I. R. Petersen are with the Research School of Engineering, The Australian National University, Canberra, Australia (e-mail:  zeinab.salehi@anu.edu.au; elizabeth.ratnam@anu.edu.au; ian.petersen@anu.edu.au).}
\thanks{Y. Chen and G. Shi are with the  Australian Center for Field Robotics, School of Aerospace, Mechanical and Mechatronic Engineering, The University of Sydney, NSW, Australia (e-mail: yijun.chen@sydney.edu.au; guodong.shi@sydney.edu.au).}
\thanks{Some preliminary results of this paper are accepted for presentation at  the 61st IEEE Conference on Decision and Control \cite{salehi2022finite},  and  the 12th IFAC Symposium on Nonlinear Control Systems  \cite{salehi2022infinite}.}
}

\maketitle

\begin{abstract}
In this paper, we consider dynamic multi-agent systems (MAS) for decentralized resource allocation. The MAS operates at a competitive equilibrium to ensure  supply and demand are balanced. First, we investigate the MAS over a finite horizon. The utility functions of agents are parameterized to incorporate individual preferences. We shape individual preferences through a set of utility functions to guarantee  the resource price at a competitive equilibrium  remains socially acceptable, i.e., the price is upper-bounded by an affordability threshold. We show this problem is solvable at the conceptual level. Next, we consider quadratic MAS and formulate the associated social shaping problem as a multi-agent linear quadratic regulator (LQR) problem which enables us to propose explicit utility sets using quadratic programming and dynamic programming. Then, a numerical algorithm is presented for calculating a tight range of the preference function parameters which guarantees a socially accepted price.  We investigate the properties of a competitive equilibrium over an infinite horizon. Considering general utility functions, we  show that under feasibility assumptions, any competitive equilibrium maximizes the social welfare. Then, we prove that for sufficiently small initial conditions, the social welfare maximization solution constitutes a competitive equilibrium with zero price. We also prove for general feasible initial conditions, there exists a time instant after which the optimal price, corresponding to a competitive equilibrium, becomes zero. Finally, we specifically focus on quadratic MAS and propose explicit results.
\end{abstract}


\section{Introduction}
\label{sec:introduction}

\IEEEPARstart{E}{fficient}  resource allocation is a fundamental problem in the literature that can be addressed by MAS approaches \cite{ananduta2021distributed, jia2021adaptive}. Resource allocation problems are of great importance when the total demand must equal the total supply for safe and secure system operation \cite{chakrabortty2011control, mallada2017optimal, liu2018stabilization}. Depending on the application, there exist two common approaches: (i) social welfare where agents collaborate to maximize the total agents' utilities \cite{singh2018decentralized,nguyen2014computational, chevaleyre}; (ii) competitive equilibrium in which agents compete to maximize their individual payoffs \cite{chen2021social, Bikhchandani1997}. 
At the competitive equilibrium, the resource price along with the allocated resources provides an effective solution to the allocation problem, thereby clearing the market \cite{wei2014competitive}.

A fundamental theorem in classical welfare economics states that the competitive equilibrium is Pareto optimal, meaning that no agent can deviate from the equilibrium to increase profits without reducing another's payoff \cite{acemoglu2018, arrow1954, peng2022pareto, zhao2020pareto}. It is also proved that under some convexity assumptions, the competitive equilibrium maximizes social welfare  \cite{mas1995microeconomic,chen2021social, Li2020}. Mechanism design is a well-known approach to social welfare maximization \cite{farhadi2018surrogate, dave2021social, ma2020incentive}. 
The key point in achieving a competitive equilibrium is efficient resource pricing that depends on the utility of each agent. The corresponding price, however, is not guaranteed to be affordable for all agents. If some participants select their utility functions aggressively, the price potentially increases to the point that it becomes unaffordable to other agents who have no alternative but to leave the system \cite{salehi2021quadratic, salehi2021social}. 
A recent example is the Texas power outage disaster in February 2021, when some citizens who had access to electricity during the period of rolling power outages received inflated (and in some cases, unaffordable) electricity bills for their daily power usage
\cite{Texas2021}.

In this paper, we consider  MAS with distributed resources and linear time-invariant (LTI) dynamics. 
%
Over a finite horizon, we achieve affordability by parameterizing utility functions and proposing some limits on the parameters such that the resource price at a competitive equilibrium never exceeds a given threshold, leading to the concept of social shaping.
We design an optimization problem and propose a conceptual scheme, based on dynamic programming, which shows  how the social shaping problem is solvable implicitly for general classes of utility functions. 
Next, the social shaping problem is reformulated for quadratic MAS, leading to an LQR problem. Solving the LQR problem using quadratic programming and dynamic programming, we propose two explicit sets for the utility function parameters such that they lead to socially acceptable resource prices.
Then, a numerical algorithm based on the bisection method is presented that provides accurate and practical bounds on the utility function parameters, followed by some convergence results.

Over an infinite horizon, we examine the behavior of resource price under a competitive equilibrium.
We extend our previous work in \cite{chen2021social} which considered a finite horizon, by showing that in the infinite horizon problem, the competitive equilibrium maximizes social welfare if feasibility assumptions are satisfied.
We also prove that for sufficiently small initial conditions, the social welfare maximization constitutes a competitive equilibrium with zero price; implying adequate resources are available in the network to meet the demand of all agents.
Furthermore, we show for any feasible initial condition,  there exists a time instant after which the optimal price, corresponding to the competitive equilibrium, is zero.
Next, as a special case, we study quadratic MAS for which the system-level social welfare optimization becomes a classical constrained linear quadratic regulator (CLQR) problem. 
Finally, we investigate how the results can be extended to tracking problems.

Some preliminary results of this paper are accepted for presentation at the 61st IEEE Conference on Decision and Control  \cite{salehi2022finite},  and   the 12th IFAC Symposium on Nonlinear Control Systems \cite{salehi2022infinite}. We discussed the problem of social shaping for dynamic MAS over a finite horizon in \cite{salehi2022finite},  where the proofs of the proposed theorems were omitted due to the page limits. In \cite{salehi2022infinite}, we discussed the existence of a competitive equilibrium and the behavior of the associated resource price over an infinite horizon. In this paper,  we extend our prior works in the following ways: (i) we introduce two real-world applications for the  problem formulation; (ii) we present the missing proofs of the proposed theorems in \cite{salehi2022finite};  (iii) we investigate tracking problems; and (iv) we provide new, and a more extensive series of numerical examples.

The remainder of this paper is organized as follows. In Section \ref{sec:preliminaries}, we introduce the problem formulation and its real-world applications. In Section \ref{sec:conceptual}, we present a conceptual scheme for solving the  social shaping problem over a finite horizon. In Section \ref{sec:Quadratic}, we address the social shaping problem for quadratic MAS over a finite horizon.  In Section \ref{sec:infinite}, we examine the existence and properties of a competitive equilibrium over an infinite horizon.  Then, we investigate how the results can be extended to tracking problems. Finally, Section \ref{sec:examples} provides numerical examples and Section \ref{sec:conclusions} contains concluding remarks.

\textit{Notation:} We denote by $\mathbb{R}$ and $\mathbb{R}^{\geq 0}$ the fields of real numbers and  non-negative real numbers, respectively. Let $\mathbf I$
denote the identity  matrix with a suitable dimension. 
The symbol $\pmb 1$ represents a vector with an appropriate dimension whose entries are all $1$. We use $\left\| \mathbf \cdot \right\|$ to denote the Euclidean norm of a vector or its induced matrix norm. Let $\sigma_{\rm min}(\cdot)$ and $\sigma_{\rm max}(\cdot)$ represent the minimum and maximum eigenvalues of a square matrix, respectively. 
\section{Problem Formulation} \label{sec:preliminaries}

\subsection{The Dynamic Multi-agent Model}\label{section:MAS}
Consider a dynamic MAS with $n$  agents indexed in the set $\mathcal V=\{ 1, 2, \dots , n\}$. This MAS is studied in the time horizon $N$. Let time steps be indexed in the set $\mathcal{T}=\{ 0, 1, \dots , N-1 \}$. Each agent $i \in \mathcal{V}$ is a subsystem with dynamics represented by
\begin{equation*}
	\mathbf x_i(t+1) = \mathbf A_i \mathbf x_i(t) + \mathbf B_i \mathbf u_i(t), \quad t \in \mathcal{T},
\end{equation*}
where $\mathbf x_i(t) \in \mathbb{R}^{d}$ is the dynamical state, $\mathbf x_i(0) \in \mathbb{R}^{d}$ is the given initial state, and $\mathbf u_i(t) \in \mathbb{R}^{m}$ is the control input. Also, $\mathbf A_i \in \mathbb{R}^{d \times d}$ and $\mathbf B_i \in \mathbb{R}^{d \times m}$ are fixed matrices. Upon reaching the state $\mathbf x_i(t)$ and employing the control input $\mathbf u_i(t)$ at time step $t \in \mathcal{T}$, each agent $i$ receives the utility $f_i(\mathbf x_i(t), \mathbf u_i(t)) = f(\cdot; \theta_i): \mathbb{R}^{d} \times \mathbb{R}^{m} \mapsto \mathbb{R}$, where $\theta_i \in \Theta$ is a personalized parameter of agent $i$. The terminal utility achieved as a result of reaching the terminal state $\mathbf x_i(N)$ is denoted by  $\phi_i(\mathbf x_i(N)) = \phi(\cdot; \theta_i): \mathbb{R}^{d} \mapsto \mathbb{R}$.
Let $a_i(t) \in \mathbb{R}$ denote the amount of excess resource generated by agent $i$ in addition to the supply they need to stay at the origin at time $t \in \mathcal{T}$. The amount of consumed resource by  the same agent for taking the control action $\mathbf u_i(t)$ is denoted by  $h_i(\mathbf u_i(t)): \mathbb{R}^{m} \mapsto \mathbb{R}^{\geq 0}$.  The total excess network supply is represented by $C(t) := \sum_{i=1}^n a_i(t)$ such that $C(t)>0$ at each time interval $t \in \mathcal{T}$.

Agents are interconnected through a network without any external resource supply. They share resources at the price $\lambda_t$, which denotes the price of unit  traded resource across the network  at time step $t \in \mathcal{T}$. The traded resource for agent $i$ is denoted by $e_i(t) \in \mathbb{R}$ which can never be greater than the net supply, i.e., $e_i(t) \leq a_i(t)- h_i(\mathbf u_i(t))$. Then, each agent $i$ receives a payoff which consists of the utility from resource consumption and the income $\lambda_t e_i(t)$ from resource exchange.
\subsection{Competitive Equilibrium}\label{sec:competitive}
Let $\mathbf U_i=(\mathbf u_i^\top(0), \dots, \mathbf u_i^\top(N-1))^\top$ and $\mathbf E_i=( e_i(0), \dots,  e_i(N-1))^\top$ denote the vector of control inputs and the vector of  trading decisions associated with agent $i$ over the whole time horizon, respectively. Also, let $\mathbf u(t)= (\mathbf u_1^\top(t), \dots, \mathbf u_n^\top(t) )^\top$ and $\mathbf e(t)= ( e_1(t), \dots,  e_n(t) )^\top$ denote the vector of  control inputs and the vector of   trading decisions associated with all agents at time step $t \in \mathcal{T}$, respectively. Let $\mathbf U=(\mathbf u^\top(0), \dots, \mathbf u^\top(N-1))^\top$ and $\mathbf E=(\mathbf e^\top(0), \dots, \mathbf e^\top(N-1))^\top$ be the vector of all control inputs and the vector of all  trading decisions at all time steps, respectively. Let $\boldsymbol  \lambda = (\lambda_0, \dots, \lambda_{N-1})^\top$ denote the vector of optimal resource prices throughout the entire time horizon.


\begin{definition}
	The \textit{competitive equilibrium } for a dynamic MAS is the triplet $(\boldsymbol{\lambda}^\ast, \mathbf{U}^\ast, \mathbf{E}^\ast)$ which satisfies the following two conditions.
	\begin{itemize}
		\item[(i)]  Given $\boldsymbol \lambda^\ast$, the pair $(\mathbf U^\ast, \mathbf E^\ast)$ maximizes the individual payoff function of each agent; i.e., each $(\mathbf U_i^\ast, \mathbf E_i^\ast)$ solves the following constrained maximization problem  
		\begin{equation}\label{opt_DLTD_1}
			\resizebox {0.447\textwidth} {!} {$
				\begin{aligned}
					\max_{{\mathbf U}_i, {\mathbf E}_i} \, \, \, &  \phi(\mathbf x_i(N); \theta_i) + \sum_{t=0}^{N-1} \Big( f( \mathbf x_i(t), \mathbf u_i(t); \theta_i)+\lambda_t^\ast e_i(t) \Big) \\
					{\rm s.t.} \quad &\mathbf x_i(t+1) = \mathbf A_i \mathbf x_i(t)+ \mathbf B_i \mathbf u_i(t),\\
					& e_i(t)  \leq a_i(t) - h_i(\mathbf u_i(t)), \quad t \in \mathcal{T}. 
				\end{aligned}$}
		\end{equation}
		\item[(ii)] The optimal  trading $\mathbf E^\ast$ balances the total traded resource across the network at each time step; that is,
		\begin{equation}\label{trading_demand_supply_constraints}
			\sum_{i=1}^n e_i^\ast(t) =0, \quad t \in \mathcal{T}. 
		\end{equation} 
	\end{itemize}
\end{definition}

\subsection{Social Welfare Maximization}
In the context of a society, the benefit of the whole community is important, not just a single agent.  It is desirable to find an operating point $ (\mathbf{U}^\ast, \mathbf{E}^\ast)$  at which  social welfare is maximized. 
The objective can be attained by solving the following social welfare maximization problem  
\begin{equation}\label{opt_social_DLD_1}
	\begin{aligned}
		\max_{\mathbf{U}, \mathbf{E}} \quad &  \sum_{i=1}^{n} \left( \phi(\mathbf x_i(N); \theta_i) + \sum_{t=0}^{N-1} f(\mathbf x_i(t), \mathbf u_i(t); \theta_i) \right) \\
		{\rm s.t.} \quad & \mathbf x_i(t+1) = \mathbf A_i \mathbf x_i(t)+ \mathbf B_i \mathbf u_i(t), \\
		& e_i(t)  \leq a_i(t) - h_i(\mathbf u_i(t)), \\
		&\sum_{i=1}^n e_i(t) =0, 
		\quad   t \in \mathcal{T}, \,\,\, i \in \mathcal V,
	\end{aligned}
\end{equation}
where the total agent utility functions are maximized.

\medskip

\begin{assumption}\label{assumption1}
	$f(\cdot; \theta_i)$ and $\phi(\cdot; \theta_i)$ are concave functions for all $i \in \mathcal{V}$. Additionally,  $h_i(\cdot)$ is a non-negative convex function such that $h_i(\mathbf z)<b$ represents a bounded open set of $\mathbf z$ in $\mathbb{R}^m$ for $b>0$ and $i \in \mathcal{V}$. Furthermore, assume $\sum_{i=1}^{n} a_i(t)>0$ for all $t \in \mathcal{T}$.
\end{assumption}

\medskip

\begin{proposition}[as in \cite{chen2021social}]\label{prop1}
	Suppose Assumption \ref{assumption1} holds. Then the competitive equilibrium and the social welfare equilibrium exist and coincide. Additionally, the optimal price $\lambda^\ast_t$ in \eqref{opt_DLTD_1} is obtained from the Lagrange multiplier associated with the balancing equality constraint $\sum_{i=1}^n e_i(t) =0$ in \eqref{opt_social_DLD_1}. 
\end{proposition}

\subsection{Social Shaping Problem}
The optimal price, $\lambda^\ast_t$, is the Lagrange multiplier corresponding to the equality constraint in \eqref{opt_social_DLD_1}, which depends on the utility functions of agents. Without regulation on the choice of utility functions, the price may become extremely high and unaffordable for some agents. In such cases, agents who find the  price unaffordable cannot compete in the market and have no alternative but to leave the system. Consequently, all of the resources will be consumed by a limited number of agents who have dominated the price by aggressively selecting their utilities --- which we consider to be socially unfair and not sustainable. To improve the affordability of the price for all agents we require a mechanism, called  social shaping,  which ensures the price is always below an acceptable threshold denoted by $\lambda^\dag \in \mathbb{R}^{>0}$. The problem of social shaping is addressed for static MAS in \cite{salehi2021quadratic} and \cite{salehi2021social}. In this paper, we define an extended version of the  social shaping problem for dynamic MAS as follows.

\begin{definition}[Social shaping for dynamic MAS]
	Consider a dynamic MAS whose agents $i \in \mathcal{V}$ have $f(\cdot; \theta_i)$ and $\phi(\cdot; \theta_i)$ as their running utility function and terminal utility function, respectively. Let $\lambda^\dag \in \mathbb{R}^{>0}$ denote a price threshold that is accepted by all agents. Find a range $\Theta$ of personal parameters $\theta_i$ such that if  $\theta_i \in \Theta$ for $i \in \mathcal{V}$,
	then we yield $\lambda_t^\ast \leq \lambda^\dag$ at all time steps $t \in \mathcal{T}$.
\end{definition}

\subsection{Real-World Applications}\label{sec:applications}
In this section, we provide two real-world applications for our problem formulation.

\medskip
\subsubsection{Community Microgrid}\label{sec:app_microgrid}
Consider a  community microgrid consisting of $n$ buildings with rooftop photovoltaic (PV)  generation. Buildings are equipped with  air conditioners to keep their temperature at a desired level, forming a group of $n$ thermostatically controlled loads (TCL). Each TCL, indexed by $i \in \mathcal{V}$, has dynamics  represented by 
$
	{ x}_i (t+1)=  A_i  x_i (t)+ B_i u_i (t)+  \tilde B_i w
$,
where  $x_i(t)$ is the internal temperature ($\tccentigrade$), $u_i(t)$ is the  consumed energy (kWh), and $ w$ reflects the impact of the ambient temperature ($\tccentigrade$)  \cite{zhao2017geometric, wang2019distributed}. 

In this context, each building $i \in \mathcal{V}$  represents an agent with excess rooftop PV energy $a_i(t)$ in (kWh) at time $t \in \mathcal{T}$. Buildings are connected through a network to share their excess energy at a price $\lambda_t^\ast$ (USD/kWh) such that the total excess PV generation balances the total demand from TCLs at each time step $t \in \mathcal{T}$. Depending on the internal temperature $x_i(t)$ and the consumed energy $u_i(t)$, each household valuates the outcome achieved through a utility function $f_i(x_i(t), u_i(t))$; for example, if the temperature is close to the desired level or the amount of consumed energy is relatively low then the household achieves a high level of satisfaction, and therefore, a high $f_i(\cdot)$.

Different choices of utility functions lead to different electricity prices $\lambda^\ast_t$. Assume that all buildings have the same class of utility functions but with different parameters; i.e., each building is associated with the utility $f_i(x_i(t), u_i(t)) = f(x_i(t), u_i(t); \theta_i)$, where $\theta_i$ is the personal parameter of building $i$ which can be selected independently. By social shaping and imposing some bounds on the choice of $\theta_i$, the coordinator guarantees that the electricity price never exceeds an acceptable threshold, 
so as to ensure affordability.

\medskip
\subsubsection{Carbon Permit Trading System}
The well-known RICE model can be formulated as a dynamic multi-agent model \cite{nordhaus1996regional, nordhaus2010economic}. Each region captured in the RICE model represents an agent. Each time step $t \in \mathcal{T}$ represents a $10$-year period. Let ${\bf x}_{i}(t)$ be each region $i$'s economic output (USD) at time step $t$ and ${\bf u}_{i}(t)$ be the amount of emission (GtCO$_2$) it emits at time step $t$. In the RICE model, each region's economic output at time step $t+1$ depends on its economic output at time step $t$ and the emitted emissions at time step $t$. Thus, each region $i \in \mathcal{V}$ has its nonlinear dynamics $\mathbf{x}_{i}(t+1) = g_{i}(\mathbf{x}_{i}(t),\mathbf{u}_{i}(t))$. Upon the states and control actions, each region $i$ evaluates its social welfare according to a utility function $f_i(\mathbf x_i(t), \mathbf u_i(t))$ capturing the fact that the societal welfare for a region depends on both the economic output and carbon emissions.

A carbon permit trading bloc scheme has been proposed using the RICE model \cite{Nordhaus2006}. Under the scheme, the total permitted global emissions at each time step $t$ is assumed to be less than $C(t)$ (GtCO$_2$). Each region $i$ is assigned with its carbon permit $a_{i}(t)$ (GtCO$_2$) at time step $t$ with the relationship $\sum_{i \in \mathcal{V}} a_{i}(t) = C(t)$. Regions are allowed to buy or sell carbon permits at a price $\lambda_{t}^{\ast}$ (USD/GtCO$_2$) such that the total carbon permits balance the total demand of emissions to be emitted by each region at each time step $t$. 

Different configurations of utility functions yield different prices $\lambda_{t}^{\ast}$ for unit carbon permit. Suppose that all regions adopt the same class of utility functions but with different choices of parameters; i.e., each region has its utility function $f(\mathbf x_i(t), \mathbf u_i(t);\theta_i)$ with $\theta_i$ being the personal parameter of region $i$. By social shaping and imposing bounds on $\theta_i$ for all $i \in \mathcal{V}$, the price per unit carbon permit is guaranteed to be below a threshold, making it affordable for all regions.
\section{Conceptual Social Shaping}\label{sec:conceptual}
In this section, we examine how the social shaping problem of dynamic MAS can be solved conceptually.

\medskip

\begin{lemma}\label{lemma1}
	Consider the dynamic MAS. If Assumption \ref{assumption1} is satisfied, then $\lambda^\ast_t \geq 0$ for all $t \in \mathcal{T}$.
\end{lemma}
\begin{proof}
	The proof is similar to the proof of Proposition 2 in \cite{chen2021social}. 
\end{proof}

\medskip 

\begin{proposition}\label{prop2}
	Consider the dynamic MAS. Let Assumption \ref{assumption1} hold. If $\lambda^\ast_t > 0$ then the total demand and supply are balanced at time step $t$; that is,
	\begin{equation}\label{eq3}
		\sum_{i=1}^{n} h_i(\mathbf u_i^\ast(t)) = C(t).
	\end{equation}
\end{proposition}
\begin{proof}
	Since Assumption \ref{assumption1} is satisfied, Proposition \ref{prop1} holds. Therefore, either the competitive problem or the social welfare problem can be solved. Consider the competitive optimization problem in \eqref{opt_DLTD_1}.  Let $s_i(t) \in \mathbb{R}^{\geq 0}$ be the slack variable for agent $i$ at time step $t \in \mathcal{T}$ and $\mathbf{S_i} = (s_i(0), ..., s_i(N-1))^\top$ be the vector of slack variables for agent $i$ throughout the whole time horizon.
	We can write the inequality constraint $e_i(t) \leq a_i(t)- h_i(\mathbf u_i(t))$ in \eqref{opt_DLTD_1} as the equality $e_i(t) = a_i(t) - h_i(\mathbf u_i(t)) -  s_i(t)$ for  $t \in \mathcal{T}$.
	Then, substituting $e_i(t)$ into \eqref{opt_DLTD_1} results in an equivalent form for the optimization problem as
	\begin{equation}\label{opt_DLTD_1_2}
			\resizebox {0.485\textwidth} {!} {$
		\begin{aligned}
			\max_{{\mathbf U}_i, {\mathbf S}_i} \quad &  \phi(\mathbf x_i(N); \theta_i) + \sum_{t=0}^{N-1}  f( \mathbf x_i(t), \mathbf u_i(t); \theta_i) \\ & + \sum_{t=0}^{N-1} \lambda_t^\ast \big[ a_i(t) - h_i(\mathbf u_i(t)) - s_i(t) \big]\\
			{\rm s.t.} \quad & \mathbf x_i(t+1) = \mathbf A_i \mathbf x_i(t)+ \mathbf B_i \mathbf u_i(t),
			\,\, s_i(t) \in \mathbb{R}^{\geq0}, \,\, t \in \mathcal{T}.
		\end{aligned}$}
	\end{equation}
	Since $\lambda^\ast_t >0$, the resulting objective function is strictly decreasing with respect to $s_i(t)$. Consequently, the optimal slack variable maximizing the objective function is $s_i^\ast(t) =0$, meaning that the associated inequality constraint is active; that is,
	$
		e_i^\ast(t) = a_i(t) - h_i(\mathbf u^\ast_i(t))
	$.
	The summation of the resulting active constraint over $i$, from $1$ to $n$, along with the balancing equality $\sum_{i=1}^{n} e_i^\ast(t)=0$ in \eqref{trading_demand_supply_constraints}, yields $\sum_{i=1}^{n}h_i(\mathbf u^\ast_i(t)) = C(t)$.
\end{proof}

\medskip

Next, lets consider the social shaping problem and an approach to solving it conceptually. Suppose Assumption \ref{assumption1} holds and $f(\cdot; \theta_i)$, $\phi(\cdot; \theta_i)$, and $h_i(\cdot)$ are continuously differentiable. Then Proposition \ref{prop1} is satisfied. In this paper, we focus on  the competitive optimization problem in \eqref{opt_DLTD_1}. 
According to Lemma \ref{lemma1}, there holds $\lambda^\ast_t \geq 0$. We can skip the case $\lambda_t^\ast=0$, because a zero price is always socially resilient. Therefore, it is sufficient to only examine $\lambda^\ast_t >0$. Following from Proposition \ref{prop2}, the total demand and supply are balanced at time step $t$, meaning that $\sum_{i=1}^{n} h_i(\mathbf u_i^\ast(t))= C(t)$; additionally, we have $s_i^\ast(t)=0$. Substituting  $s_i(t) =0$ into \eqref{opt_DLTD_1_2} yields an equivalent form for the competitive optimization problem in \eqref{opt_DLTD_1} as
\begin{equation}\label{opt_DLTD_1_3}
	\begin{aligned}
		\max_{{\mathbf U}_i} \quad &  \phi(\mathbf x_i(N); \theta_i ) + \sum_{t=0}^{N-1} \Big( f( \mathbf x_i(t), \mathbf u_i(t); \theta_i )
		\\ &+  \lambda_t^\ast \big[ a_i(t) - h_i(\mathbf u_i(t)) \big] \Big)\\
		{\rm s.t.} \quad & \mathbf x_i(t+1) = \mathbf A_i \mathbf x_i(t)+ \mathbf B_i \mathbf u_i(t), \quad t \in \mathcal{T},
	\end{aligned}
\end{equation}
which is valid for $\lambda^\ast_t>0$. Please note that even if $\lambda^\ast_t = 0$ then \eqref{opt_DLTD_1_2} can be written in the form of \eqref{opt_DLTD_1_3}, although the equality in \eqref{eq3} turns into the inequality $\sum_{i=1}^{n}h_i(\mathbf u^\ast_i(t)) \leq C(t)$. This fact causes no change to the upcoming analysis. 

The optimization problem in \eqref{opt_DLTD_1_3} is an unconstrained optimal control problem 
which can be solved with dynamic programming. First, introduce the cost-to-go function for agent $i$ from time $k$ to $N$ as
\begin{equation*}
	\resizebox {0.485\textwidth} {!} {$
		\begin{aligned}
			&J_{i}^{k \longrightarrow N} (\mathbf x_i(k), \mathbf u_i(k),  \dots, \mathbf u_i(N-1), \lambda^\ast_k, \dots, \lambda^\ast_{N-1}; \theta_i ) \\ &= \phi ( \mathbf x_i(N); \theta_i ) + \sum_{t=k}^{N-1} \Big( f (\mathbf x_i(t), \mathbf u_i(t); \theta_i )+\lambda_t^\ast \big[ a_i(t) - h_i (\mathbf u_i(t)) \big]\Big).
		\end{aligned} $}
\end{equation*}
Then, the optimal cost-to-go at time $k$ for agent $i$, which is also called the value function, is represented as
\begin{equation*}
	\resizebox {0.485\textwidth} {!} {$
		\begin{aligned}
			&V_{i, k}(\mathbf x_i(k), \lambda^\ast_k, \dots, \lambda^\ast_{N-1}; \theta_i) \\ &= \max_{\mathbf u_i(k), \dots, \mathbf u_i(N-1)} \quad   J_{i}^{k \longrightarrow N} (\mathbf x_i(k), \mathbf u_i(k), \dots, \mathbf u_i(N-1), \lambda^\ast_k, \dots, \lambda^\ast_{N-1}; \theta_i) \\
			& \quad \quad \quad \quad  {\rm s.t.} \, \, \,  \quad \quad \quad  \mathbf x_i(t+1) = \mathbf A_i \mathbf x_i(t)+ \mathbf B_i \mathbf u_i(t), \quad t= k, \dots, N-1.
		\end{aligned} $}
\end{equation*}
According to the principle of optimality, we obtain
\begin{equation*}\label{eq17}
	\begin{aligned}
		&V_{i, N}(\mathbf x_i(N); \theta_i) = \phi \big( \mathbf x_i(N); \theta_i \big),\\
		&V_{i, N-1}(\mathbf x_i(N-1), \lambda^\ast_{N-1}; \theta_i) \\&= \max_{\mathbf u_i(N-1)} \quad   f (\mathbf x_i(N-1), \mathbf u_i(N-1); \theta_i) \\& \quad \quad \quad \quad \quad +\lambda_{N-1}^\ast \big[ a_i(N-1) - h_i \big(\mathbf u_i(N-1)\big) \big]  \\& \quad \quad \quad \quad \quad + V_{i, N}(\mathbf A_i \mathbf x_i(N-1)+ \mathbf B_i \mathbf u_i(N-1); \theta_i),\\
		\vdots\\
		&V_{i, 0}(\mathbf x_i(0), \lambda_{0}^\ast, \dots, \lambda_{N-1}^\ast; \theta_i) \\ &= \max_{\mathbf u_i(0)} \quad   f (\mathbf x_i(0), \mathbf u_i(0); \theta_i)+ \lambda_{0}^\ast \big[ a_i(0) - h_i \big(\mathbf u_i(0)\big) \big] \\ & \quad \quad \quad \quad \quad + V_{i, 1}(\mathbf A_i \mathbf x_i(0)+ \mathbf B_i \mathbf u_i(0), \lambda_{1}^\ast, \dots, \lambda_{N-1}^\ast; \theta_i).
	\end{aligned}
\end{equation*}
To obtain the optimal control at time step $k=0$, the derivative of the associated objective function with respect to $\mathbf u_i(0)$ must equal zero; that is,
\begin{multline*}
	\frac{\partial f(\mathbf x_i(0), \mathbf u_i(0);\theta_i)}{\partial \mathbf u_i(0)}- \lambda_{0}^\ast \nabla h_i \big(\mathbf u_i(0)\big)\\ + \frac{\partial V_{i,1}(\mathbf A_i\mathbf x_i(0) +\mathbf B_i \mathbf u_i(0), \lambda_{1}^\ast, \dots, \lambda_{N-1}^\ast;\theta_i)}{\partial \mathbf u_i(0)}=0.
\end{multline*} 
Proposition \ref{prop1} implies that such an optimal solution exists, although it might not be unique. {Without loss of generality, suppose the optimal solution is unique.} Thus, we can write $\mathbf u_i^\ast(0)$ as a function of $\mathbf x_i(0)$ and all $\lambda^\ast_t$ where $t \in \mathcal{T}$,   parameterized by $\theta_i$; that is,
\begin{equation}\label{eq20}
	\mathbf u_i^\ast(0)=l_i^0(\mathbf x_i(0), \lambda_{0}^\ast, \dots, \lambda_{N-1}^\ast; \theta_i).
\end{equation}
Substituting \eqref{eq20} into the  equality $\sum_{i=1}^{n} h_i(\mathbf u_i^\ast(t)) = C(t)$ in \eqref{eq3}, we yield
$
	\sum_{i=1}^{n} h_i(l_i^0(\mathbf x_i(0), \lambda_{0}^\ast, \dots, \lambda_{N-1}^\ast; \theta_i)) = C(0)
$.
Similarly, for any other time step $k \in \mathcal{T}$ we  achieve
$
\mathbf u_i^\ast(k)= l_i^k(\mathbf x_i(0), \lambda^\ast_0, \dots, \lambda_{N-1}^\ast; \theta_i)
$,
and
\begin{equation}\label{eq32}
	\sum_{i=1}^{n} h_i( l_i^k(\mathbf x_i(0), \lambda^\ast_0, \dots, \lambda_{N-1}^\ast; \theta_i)) = C(k), \,\,\,\,  k \in \mathcal{T}. 
\end{equation}
We aim to obtain $\boldsymbol{\lambda^\ast} = (\lambda^\ast_0, \dots, \lambda^\ast_{N-1})^\top$. According to \eqref{eq32}, we have $N$ equations with $N$ variables. 
Let $\mathbf x(k) = (\mathbf x^\top_1(k), \dots, \mathbf x^\top_n(k))^\top$, $\boldsymbol{\theta}=(\theta_1, \theta_2, \dots, \theta_n)$, and $\mathbf C =(C(0), C(1), \dots, C(N-1))$. According to Proposition \ref{prop1}, there exists $\boldsymbol \lambda^\ast$ which satisfies \eqref{eq32} although it might not be unique. Among different possible prices that satisfy the equilibrium, we consider the maximum one at each time step. In the rest of this proof, by optimal price we mean the maximum possible price associated with a fixed $\boldsymbol{\theta}$ meeting the equilibrium conditions. Solving \eqref{eq32}, the optimal price at each time step $k \in \mathcal{T}$ is obtained as
$
\lambda^\ast_k= g_k(\mathbf x(0), \mathbf C; \boldsymbol \theta)
$.
Additionally, for different values of agent preferences $\boldsymbol \theta$ we would obtain different optimal prices at each time step. Let us define the maximum value of the set of all possible optimal prices at each time step $k \in \mathcal{T}$, when $\theta_i$ takes values in the set $\Theta$ (or $\boldsymbol{\theta} \in \Theta^n$), as
$$
\chi^\Theta_k : = \max_{\boldsymbol \theta \in \Theta^n} g_k(\cdot; \boldsymbol \theta), \quad k=0, \dots, N-1.
$$
Next, we introduce
$
\mathbf G_\Theta:=\big(	\chi^\Theta_0, \chi^\Theta_1, \dots, \chi^\Theta_{N-1} \big)^\top
$.
Each element $k$ in the vector $\mathbf G_\Theta$ is the maximum value of optimal prices at  time step $k$, when agent preferences are taken from $\boldsymbol \theta \in \Theta^n$.  This leads to the following result.

\medskip

\begin{theorem}
Consider a dynamic MAS. Let Assumption \ref{assumption1} hold. 
Suppose $f(\cdot; \theta_i)$, $\phi(\cdot; \theta_i)$, and $h_i(\cdot)$ are continuously differentiable. Let $\lambda^\dag \in \mathbb{R}^{>0}$ represent the given price threshold accepted by all agents.  
Then any set $\Theta$ satisfying $\mathbf G_\Theta \leq \lambda^\dag \pmb{1}$ ensures that $\lambda^\ast_t \leq \lambda^\dag$ for $t \in \mathcal{T}$, and thus, solves the social shaping problem of agent preferences. 
\end{theorem}

\section{Quadratic Social Shaping}\label{sec:Quadratic}
In this section, we  examine quadratic utility functions and explicitly propose two sets of personal parameters that lead to socially acceptable optimal prices.

\medskip

\begin{assumption}\label{assumption2}
Consider the dynamic MAS introduced in Section \ref{section:MAS}. Let $\theta_i:= (\mathbf Q_i, \mathbf R_i)$, where $\mathbf Q_i \in \mathbb{R}^{d \times d}$, $\mathbf Q_i=\mathbf Q_i^\top>0$ and $\mathbf R_i \in \mathbb{R}^{m \times m}$, $\mathbf R_i=\mathbf R_i^\top>0$.
Assume for all $i\in \mathcal V$ we have	
\begin{equation*}\label{eq_utility_q}
	\begin{gathered}
		\begin{aligned}	
			f(\mathbf x_i(t), \mathbf u_i(t); \theta_i) &= - \mathbf x_i^\top(t) \mathbf Q_i \mathbf x_i(t)  - \mathbf u_i^\top(t) \mathbf R_i \mathbf u_i(t),\\
			\phi(\mathbf x_i(N); \theta_i) &= - \mathbf x_i^\top(N) \mathbf Q_i \mathbf x_i(N),\\
			h_i(\mathbf u_i(t)) &=  \mathbf u_i^\top(t)\mathbf H_i \mathbf u_i(t),
		\end{aligned}
	\end{gathered}
\end{equation*}
where $\mathbf H_i \in \mathbb{R}^{m \times m}$, $\mathbf H_i=\mathbf H_i^\top>0$. 
\end{assumption}

\medskip 

\begin{assumption}\label{assumption3}
Consider the dynamic MAS in Assumption \ref{assumption2} with a given initial state $\mathbf x_i(0)$ such that  $\left\| {{\mathbf x_i}(0)} \right\| \le \gamma$, $\left\|\mathbf A_i \right\| \leq \alpha$, $\left\| \mathbf B_i \right\| \leq \beta$, and $\mathbf H_i \geq \rho \mathbf I$ for $i \in \mathcal V$. Suppose that $\gamma, \alpha, \beta, \rho \in \mathbb{R}^{>0}$.
\end{assumption}

\medskip 

We aim to solve the following social shaping problem.

\medskip

\noindent{\bf Dynamic \& Quadratic Social Shaping Problem.} Suppose Assumptions \ref{assumption2} and \ref{assumption3} hold. Let $\lambda^\dag \in \mathbb{R}^{>0}$ be the given price threshold accepted by all agents, and $\delta_{\rm max}  \in \mathbb{R}^{>0}$ be an upper bound for the norm of the personal parameter $\mathbf Q_i$. We propose an admissible set for $\delta_{\rm max}$ such that all utility functions satisfying $\left\| \mathbf Q_i \right\| \leq \delta_{\rm max}$  (or $\mathbf Q_i \leq \delta_{\rm max} \mathbf I$) lead to socially acceptable energy prices at all time steps; i.e., $ \lambda^\ast_t \leq \lambda^\dag$ for $t \in \mathcal{T}$. 

\medskip

To address this problem, we use two approaches:  quadratic programming and dynamic programming.

\subsection{Quadratic Programming Approach}\label{sec:QuadProg}
Since Assumption \ref{assumption2} is satisfied, Proposition \ref{prop1} holds. We examine the competitive optimization problem in \eqref{opt_DLTD_1}. According to Lemma \ref{lemma1}, there holds $\lambda^\ast_t \geq 0$. We skip the case $\lambda_t^\ast=0$, because a zero price is always socially resilient. Hence, it is sufficient to only study $\lambda^\ast_t >0$.  According to \eqref{opt_DLTD_1_3}, the optimization problem in \eqref{opt_DLTD_1} can be reformulated as 
\begin{equation}\label{opt_DLTD_2}
\begin{aligned}
	\max_{{\mathbf U}_i} \quad &- \mathbf x_i^\top(N) \mathbf Q_i \mathbf x_i(N) + \sum_{t=0}^{N-1} \Big[ \big. - \mathbf x_i^\top(t) \mathbf Q_i \mathbf x_i(t)   \\  &-\mathbf u_i^\top(t) \mathbf R_i \mathbf u_i(t)  +  \lambda_t^\ast \big(a_i(t) -  \mathbf u_i^\top(t)\mathbf H_i \mathbf u_i(t) \big) \big. \Big]  \\
	{\rm s.t.} \quad & \mathbf x_i(t+1) = \mathbf A_i \mathbf x_i(t)+ \mathbf B_i \mathbf u_i(t), \quad   t \in \mathcal{T}. \\ 
\end{aligned}
\end{equation}

\begin{theorem}\label{theorem2}
Consider the dynamic MAS described in Assumptions \ref{assumption2} and \ref{assumption3} on the time horizon $N$.
Suppose $\delta_{\rm max} \in \mathbb{R}^{>0}$ is selected from the following set
\begin{equation*}\label{eq_set_Q1}
			\begin{aligned}
				&\mathscr{S}_{\ast} =\left\{ \Bigg. \right. \delta_{\rm max} \in \mathbb{R}^{>0} : 
				\delta_{\rm max} \sum_{t=k+1}^{N}  \Bigg[ \Bigg. \gamma \alpha^{2t-k-1}   \\ &+ \beta   { \sum_{\substack{ j=0 \\ j\ne k}}^{t-1}} \sqrt{\frac{C(j)}{\rho}} \alpha^{2t-j-k-2}  \Bigg.\Bigg]  \leq \frac{\sqrt{C(k)\rho}}{n \beta}\lambda^\dag \, \, \, \, {\rm for}\,\, \forall k \in \mathcal{T} \left. \Bigg. \right\}.
			\end{aligned}
\end{equation*} 
Then for all quadratic utility functions satisfying $\left\| \mathbf Q_i \right\| \leq \delta_{\rm max}$  (or $\mathbf Q_i \leq \delta_{\rm max} \mathbf I$), the resulting optimal price is socially resilient, i.e., $\boldsymbol \lambda^\ast \leq \lambda^\dag \pmb{1}$.
\end{theorem}	
\begin{proof}
Considering the  equality in \eqref{eq3}, we obtain
\begin{equation}\label{eq3_2}
	\sum_{i=1}^{n} \mathbf u_i^{\ast \top}(t) \mathbf H_i \mathbf u_i^\ast(t) = C(t), \quad t \in \mathcal{T}.
\end{equation}
Furthermore, the following inequality holds
\begin{equation}\label{eq3_3}
	\sigma_{\rm min}(\mathbf H_i) \left\| \mathbf u_i^\ast(t) \right\|^2 \leq \mathbf u_i^{\ast \top}(t) \mathbf H_i \mathbf u_i^\ast(t).
\end{equation}
Additionally, since $\mathbf u_i^{\ast \top}(t) \mathbf H_i \mathbf u_i^\ast(t) \geq 0$, the equality in \eqref{eq3_2} yields 
	$	\mathbf u_i^{\ast \top}(t) \mathbf H_i \mathbf u_i^\ast(t) \leq C(t).$
Following from \eqref{eq3_3}, we obtain
\begin{equation}\label{eq3_5}
	\sigma_{\rm min}(\mathbf H_i) \left\| \mathbf u_i^\ast(t) \right\|^2 \leq C(t).
\end{equation}
According to Assumption \ref{assumption3}, we have $\mathbf H_i \geq \rho I$, meaning that $\sigma_{\rm min}(\mathbf H_i) \geq \rho$. Consequently, inequality \eqref{eq3_5} results in 
\begin{equation}\label{eq3_6}
	\left\|\mathbf u_i^\ast(t) \right\| \leq \sqrt{\frac{C(t)}{\rho}}.
\end{equation}
Additionally, from the dynamical equation in \eqref{opt_DLTD_2}, we obtain
\begin{equation}\label{eq15}
	\mathbf x_i(t)=\mathbf A_i^t \mathbf x_i(0)+ \sum_{j=0}^{t-1}{\mathbf A_i^{t-j-1} \mathbf B_i \mathbf u_i(j)}, \quad t \in \{1, 2, ..., N\}.
\end{equation}
Substituting \eqref{eq15} into \eqref{opt_DLTD_2} yields an unconstrained optimization problem, in which the only decision variable is ${\mathbf U}_i$. The associated objective function $J$ is 
\begin{equation*}\label{opt_DLTD_4}
		\small
		\begin{split}
			&J :=  \sum_{t=1}^{N} \Bigg[ \Bigg. - \Bigg( \mathbf A_i^t \mathbf x_i(0) + \sum_{j=0}^{t-1} \mathbf A_i^{t-j-1} \mathbf B_i \mathbf u_i(j) \Bigg)^\top \mathbf Q_i \\ &\times \Bigg( \mathbf A_i^t \mathbf x_i(0) + \sum_{j=0}^{t-1} \mathbf A_i^{t-j-1} \mathbf B_i \mathbf u_i(j) \Bigg)  \Bigg.\Bigg]  - \mathbf x_i^\top(0) \mathbf Q_i \mathbf x_i(0)   \\   &+\sum_{t=0}^{N-1} \Bigg[ \Bigg. - \mathbf u_i^\top(t) \mathbf R_i \mathbf u_i(t) +\lambda_t^\ast \Bigg(  a_i(t) -  \mathbf u_i^\top(t) \mathbf H_i \mathbf u_i(t) \Bigg) \Bigg.\Bigg]. 
		\end{split} 
\end{equation*}
Let $k$ denote a time interval indexed in $\mathcal{T}$. 
Setting $\frac{\partial J}{\partial \mathbf u_i(k)} =0$, applying some matrix manipulations, and using  \eqref{eq3_2}, we obtain
\begin{equation}\label{eq25}
	\small
	\resizebox {0.487\textwidth} {!} {$
	\begin{aligned}
			&\lambda^\ast_k = -\frac{1}{C(k)}  \sum_{t=k+1}^{N} \Bigg[ \sum_{i=1}^{n}  \Bigg.  \Bigg(\mathbf A_i^{t-k-1} \mathbf B_i \mathbf u_i(k) \Bigg)^\top \mathbf Q_i \Bigg( \mathbf A_i^t \mathbf x_i(0) \\&+ \sum_{j=0}^{t-1} \mathbf A_i^{t-j-1} \mathbf B_i \mathbf u_i(j) \Bigg)  \Bigg.\Bigg]- \frac{1}{C(k)}\sum_{i=1}^{n} \mathbf u_i^\top(k) \mathbf R_i \mathbf u_i(k).
	\end{aligned}$}
\end{equation}	
Since $\mathbf Q_i>0$, we obtain $$-(\mathbf A_i^{t-k-1}\mathbf B_i \mathbf u_i(k))^\top \mathbf Q_i \mathbf A_i^{t-k-1} \mathbf B_i \mathbf u_i(k)\leq 0.$$ Similarly, $\mathbf R_i>0$, and  $- \mathbf u_i^\top(k) \mathbf R_i \mathbf u_i(k) \leq 0$.  
Next, we seek an upper bound for $\lambda^\ast_k$. First, let these two terms equal zero. Then use Assumption \ref{assumption3} and the inequality in \eqref{eq3_6}, and by substitution into the norm of  \eqref{eq25},  we yield
\begin{multline}\label{eq24}
	\lambda^\ast_k \leq \frac{1}{C(k)}  \sum_{t=k+1}^{N}  \Bigg[ \Bigg. n \alpha^{2t-k-1}\beta \sqrt{\frac{C(k)}{\rho}}  \delta_{\rm max} \gamma  \\ + n \beta^2 \sqrt{\frac{C(k)}{\rho}} \delta_{\rm max}\sum_{\substack{ j=0 \\ j\ne k}}^{t-1} \alpha^{2t-j-k-2} \sqrt{\frac{C(j)}{\rho}} \Bigg.\Bigg].
\end{multline}
By assumption, the right-hand side of \eqref{eq24} is less than or equal to $\lambda^\dag$. Therefore, we obtain $\lambda^\ast_k \leq \lambda^\dag$.
\end{proof}
					
\subsection{Dynamic Programming Approach}
Similar to Section \ref{sec:QuadProg}, let $\lambda^\ast_t >0$ and consider the optimization problem in \eqref{opt_DLTD_2} which is equivalent to
\begin{equation}\label{opt_DLTD_2_3}
	\begin{aligned}
		\max_{{\mathbf U}_i} \quad   &- \mathbf x_i^\top(N) \mathbf Q_i \mathbf x_i(N) + \sum_{t=0}^{N-1} \Big[ \big. - \mathbf x_i^\top(t) \mathbf Q_i \mathbf x_i(t) \\ &- \mathbf u_i^\top(t) \bigg(\mathbf R_i + \lambda^\ast_t \mathbf H_i \bigg) \mathbf u_i(t)+\lambda_t^\ast  a_i(t) \big. \Big] \\
		{\rm s.t.} \quad & \mathbf x_i(t+1) = \mathbf A_i \mathbf x_i(t)+ \mathbf B_i \mathbf u_i(t), \quad t \in \mathcal{T}.
		\end{aligned}
\end{equation}
						
\medskip
						
\begin{theorem}\label{theorem3}
	Consider the dynamic MAS on the time horizon ${N}$. Let Assumptions \ref{assumption2} and \ref{assumption3} hold. Suppose $\delta_{\rm max} \in \mathbb{R}^{>0}$ is selected from the following set
	\begin{multline*}\label{eq_set_Q2}
	\mathscr{S}_{\ast} =\left\{ \Bigg. \right. \delta_{\rm max} \in \mathbb{R}^{>0} : 
	\delta_{\rm max} \sum_{t=1}^{N}  \gamma \alpha^{2t-1}    \leq \frac{\sqrt{C(0)\rho}}{n \beta}  \lambda^\dag,\\
	\delta_{\rm max} \sum_{t=k+1}^{N}  \Bigg[ \Bigg. \gamma \alpha^{2t-k-1}   + \beta  {\sum_{\substack{ j=0 }}^{k-1}} \sqrt{\frac{C(j)}{\rho}} \alpha^{2t-j-k-2} \Bigg.\Bigg] \\  \leq \frac{\sqrt{C(k)\rho}}{n \beta}  \lambda^\dag \, \, \, \,  {\rm for}\,\, \forall k \in \mathcal{T}, k\neq 0\left. \Bigg. \right\}.
\end{multline*} 
Then, the resulting  $\boldsymbol \lambda^\ast$ is socially resilient for all utility functions satisfying $\left\| \mathbf Q_i \right\| \leq \delta_{\rm max}$  (or $\mathbf Q_i \leq \delta_{\rm max} \mathbf I$).
\end{theorem}
						
\begin{proof}
Similar to the previous section,  \eqref{eq3_2} and \eqref{eq3_6} are satisfied. 
Since $\lambda^\ast_t > 0$, we obtain $ \mathbf R_i + \lambda^\ast_t \mathbf H_i >0$. Consequently, the optimization problem in \eqref{opt_DLTD_2_3} is a \emph{standard LQR problem}. Therefore, at each time step $k \in \mathcal{T}$, the optimal control solution is obtained as
$	\mathbf u_i^\ast(k) = - \big(\mathbf B_i^\top \mathbf P_{i, k+1} \mathbf B_i + ( \mathbf R_i + \lambda^\ast_k \mathbf H_i) \big)^{-1} \mathbf B_i^\top \mathbf P_{i, k+1} \mathbf A_i \mathbf x_i(k),$
where 
\begin{equation}\label{eq_p}
	\begin{aligned}
	\mathbf P_{i,k} =\mathbf A_i^\top \mathbf P_{i,k+1} \mathbf A_i + \mathbf Q_i - \mathbf A_i^\top \mathbf P_{i,k+1} \mathbf B_i \big( \mathbf B_i^\top \mathbf P_{i, k+1} \mathbf B_i \\ + (\mathbf R_i + \lambda^\ast_k \mathbf H_i) \big)^{-1} \mathbf B_i^\top \mathbf P_{i, k+1} \mathbf A_i.
	\end{aligned}
\end{equation}
The Riccati difference equation in \eqref{eq_p} is initialized with $\mathbf P_{i,N}= \mathbf Q_i$ and is solved backward from $k=N-1$ to $k=0$. Additionally, since the last term on the right-hand side of \eqref{eq_p} is negative semi-definite, we obtain
$
	\mathbf P_{i,k} \leq \mathbf A_i^\top \mathbf P_{i,k+1} \mathbf A_i + \mathbf Q_i
$,
and therefore,
$
	\left\|\mathbf P_{i,k} \right\| \leq \alpha^2 \left\| \mathbf P_{i,k+1} \right\| + \left\| \mathbf Q_i \right\|
$.
Starting from $k=N-1$, we obtain
\begin{equation}\label{eq37}
	\begin{aligned}
		\left\|\mathbf P_{i,N-1} \right\| \leq & (\alpha^2 +1 ) \left\| \mathbf Q_i \right\|, \\
		\vdots\\
		\left\|\mathbf P_{i,N-p} \right\| \leq & (\alpha^{2p} + \alpha^{2(p-1)}+ \dots+\alpha^2 +1 ) \left\| \mathbf Q_i \right\|.	
  \end{aligned}
\end{equation}
To find the optimal control input, we start from $k=0$ and proceed in a forward manner. In the first step, we obtain
\begin{equation}\label{eq34}
	\mathbf u_i^\ast(0) = - \big(\mathbf B_i^\top \mathbf P_{i, 1} \mathbf B_i +  \mathbf R_i + \lambda^\ast_0 \mathbf H_i \big)^{-1} \mathbf B_i^\top \mathbf P_{i, 1} \mathbf A_i \mathbf x_i(0).
\end{equation}
Applying some matrix manipulations and using the equality in  \eqref{eq3_2}, we extract $\lambda^\ast_0$ from \eqref{eq34} as	\begin{multline}\label{eq40}
	\small
	\lambda_{0}^\ast = - \frac{1}{C(0)} \sum_{i=1}^{n}\mathbf u_i^{\ast \top}(0) \mathbf B_i^\top \mathbf P_{i, 1} \mathbf A_i \mathbf x_i(0) \\- \frac{1}{C(0)}\sum_{i=1}^{n} \mathbf u_i^{\ast \top}(0)\big(\mathbf B_i^\top \mathbf P_{i, 1} \mathbf B_i +  \mathbf R_i \big) \mathbf u_i^\ast(0).
\end{multline}
To obtain an upper bound for $\lambda^\ast_0$, we omit the second term on the right-hand side of \eqref{eq40} which is always non-positive. 
Additionally,  using norm properties and considering Assumption \ref{assumption3} and the inequalities in \eqref{eq3_6} and \eqref{eq37}, we obtain 
\begin{equation}\label{eq41}
	\begin{aligned}
		\lambda_{0}^\ast & \leq  \frac{n}{C(0)} \sqrt{\frac{C(0)}{\rho}} \beta  \gamma \delta_{\rm max} \sum_{t=1}^{N} \alpha^{2t-1}.
	\end{aligned}
\end{equation}
Moving forward to the time step $k>0$, we obtain
\begin{multline}\label{eq48}
	\lambda^\ast_k \leq \frac{n}{\sqrt{C(k)\rho}} \beta \delta_{\rm max} \sum_{t=k+1}^{N}  \Bigg[ \Bigg. \gamma \alpha^{2t-k-1}   \\+ \beta   \sum_{\substack{ j=0 }}^{k-1} \sqrt{\frac{C(j)}{\rho}} \alpha^{2t-j-k-2} \Bigg.\Bigg]. 
\end{multline}
	By assumption, the right-hand side of \eqref{eq41} and \eqref{eq48} is less than or equal to $\lambda^\dag$, which confirms $\lambda^\ast_k \leq \lambda^\dag$ for $k \in \mathcal{T}$.
\end{proof}

										
\subsection{Numerical Algorithm}\label{sec:NumericalAlgorithm}
The two proposed sets in Theorems \ref{theorem2} and \ref{theorem3} are conservative but provide insight into the trade-off between utility functions' parameters and achieving the price threshold. To obtain more accurate and practical results, we propose a numerical algorithm  which provides less conservative bounds on the  parameters. The algorithm proceeds based on the bisection method.
										
\medskip
										
\noindent\textbf{Numerical Social Shaping Problem.} Consider the social welfare problem in \eqref{opt_social_DLD_1}. Let $\delta_{\rm max} \in \mathbb{R}^{>0}$ denote the design parameter. Suppose Assumption \ref{assumption2} holds with  $\mathbf Q_i = q_i \mathbf I$ where agents have the freedom to select $q_i \in \left(0, \delta_{\rm max} \right]$. Assume $\lambda^\dag$ is the  price threshold accepted by all agents and $\mathbf R_i$ is specified for each $i \in \mathcal{V}$. We aim to find the upper bound $\delta_{\rm max}$ by a numerical approach such that if $q_i \in \left(0, \delta_{\rm max} \right]$ for $i \in \mathcal{V}$ then $\lambda^\ast_t \leq \lambda^\dag$ for $t \in \mathcal{T}$. Accordingly, the key steps required are illustrated in Algorithm \ref{algorithm1}.
										
										
\begin{algorithm}[h]
\SetAlgoLined
\textbf{Input:} {System parameters $\mathbf A_i$, $\mathbf B_i$, and $\mathbf H_i$, the initial state $\mathbf x_i(0)$,  the time horizon $N$, the penalty matrix $\mathbf R_i$, and the local supply $a_i(t)$ for $i \in \mathcal{V}$ and $t \in \mathcal{T}$}.
											
\textbf{Structure:} {Consider $\mathbf Q_i = q_i \mathbf I$}. Define 
\begin{equation}\label{eq52}
	\bar \lambda^\ast(\delta) = \max_{q_1, \dots, q_n \in \left(0, \delta \right] }\max_{t \in \mathcal{T}} \lambda^\ast_t.
\end{equation}
											
\textbf{Initialize:} Set $k=0$, $b_0 = 0$, and $d_0 =d_\varrho>0$ such that $d_\varrho$ is sufficiently large to satisfy $\bar \lambda^\ast(d_\varrho)>\lambda^\dag$. 
											
\While{\rm True}{
	$L_k=(b_k + d_k)/2$, \quad
	$\lambda_k=\bar \lambda^\ast(L_k)$;
												
	\uIf{$\lambda_k > \lambda^\dag$}{
	$b_{k+1} = b_k$ and $d_{k+1}=L_k$;\\
	$k=k+1$;
}
\uElseIf{$\lambda_k < \lambda^\dag$}{
$b_{k+1} = L_k$  and $d_{k+1}=d_k$;\\
$k=k+1$;
}
\Else{
	$\delta_{\rm max} = L_k$;\\
	{break}
}
												
}	
\textbf{Output:} {$\delta_{\rm max}=L_k$} if the algorithm stops after a finite number of steps. Otherwise, $\delta_{\rm max}=\lim_{k \rightarrow \infty}L_k$.
											
\caption{Bisection-Based  Social Shaping}
\label{algorithm1}
\end{algorithm}
										
\begin{lemma}\label{lemma2}
The function $\bar \lambda^\ast(\delta)$ in \eqref{eq52} is monotonically increasing.
\end{lemma}
\begin{proof}
	When $\delta$ increases, the domain of agents' preferences expands respectively.  Therefore, the maximum possible price can never decrease.
\end{proof}
										
\medskip
										
\begin{theorem}\label{theorem7}
	The auxiliary variable $L_k$ in Algorithm \ref{algorithm1} converges to $L^\ast$ for some $L^\ast \in \left(0, d_\varrho \right)$ when $k \rightarrow \infty$. 
\end{theorem}
\begin{proof}
	From the update rules in Algorithm \ref{algorithm1}, we obtain 
	\begin{equation}\label{eq54}
	b_{k+1} \geq b_k, \quad \quad d_{k+1} \leq d_k,
	\end{equation} 
	and
	\begin{equation}\label{eq51}
    	b_0\leq b_k < L_k < d_k \leq d_0.
	\end{equation}
	Inequalities in \eqref{eq54} and \eqref{eq51} imply that $b_k$ is monotonically increasing and  bounded above by $d_0=d_\varrho$. Similarly, $d_k$ is monotonically decreasing and  bounded below by $b_0=0$. Consequently, $b_k$ and $d_k$ converge when $k \rightarrow \infty$.
											
	Additionally, from the algorithmic steps, we obtain
	$
		| b_{k} - d_{k}| = 0.5^k d_\varrho.
	$
	Therefore, 
	\begin{equation}\label{eq50}
		\lim_{k \rightarrow \infty} |b_k - d_k| =0.
	\end{equation}
	Considering \eqref{eq54}, \eqref{eq51}, and \eqref{eq50}, we conclude 
	\begin{equation}\label{eq53}
		\lim_{k \rightarrow \infty} b_k=\lim_{k \rightarrow \infty} d_k= 	\lim_{k \rightarrow \infty} L_k =L^\ast,
	\end{equation}
		where $b_k < L^\ast< d_k$ at each iteration, and $L^\ast \in \left(0, d_\varrho \right)$.
		%
\end{proof}
										
\medskip

\begin{theorem}\label{theorem4}
Suppose there exists $\delta^\dag>0$ such that $\bar \lambda^\ast(\delta^\dag)=\lambda^\dag$. Then $\delta_{\rm max}$ obtained from Algorithm \ref{algorithm1} satisfies $\bar\lambda^\ast(\delta_{\rm max}) = \lambda^\dag$.
\end{theorem}
\begin{proof}
	If the algorithm stops  when $\lambda_k = \lambda^\dag$, we end with $\delta_{\rm max} = L_k$ and  $\bar \lambda^\ast (\delta_{\rm max}) = \lambda^\dag$. Otherwise, we obtain $\delta_{\rm max} =  \lim_{k \rightarrow \infty}L_k=L^\ast$.
	By contradiction, suppose $\bar \lambda^\ast(L^\ast) \neq \lambda^\dag$ leading to $L^\ast  \neq \delta^\dag$. First, consider $L^\ast < \delta^\dag$. According to  \eqref{eq53}, we obtain
	$$\exists  \, l>0 \quad {\rm s.t.} \quad d_k < \delta^\dag \quad {\rm for} \quad \forall \, k \geq l.$$
	Following Lemma \ref{lemma2}, we obtain $\bar \lambda^\ast(d_k) \leq \bar \lambda^\ast(\delta^\dag) = \lambda^\dag$, which contradicts the assumptions in Algorithm \ref{algorithm1}.
	If $L^\ast > \delta^\dag$, a similar analysis can be made for $b_k$ which leads to a contradiction. Consequently, it follows that $\bar \lambda^\ast(L^\ast) = \bar \lambda^\ast(\delta_{\rm max}) = \lambda^\dag$. 
\end{proof}
\section{Extending the Horizon to Infinity}\label{sec:infinite}
In this section, we investigate the properties of a competitive equilibrium, especially the resource price, over an infinite horizon. 

Similar to the finite horizon case (see sections \ref{section:MAS} and \ref{sec:competitive}), define  $\mathcal{T}$, $\boldsymbol  \lambda $,  $\mathbf U_i$, $\mathbf E_i$, $\mathbf U$, and $\mathbf E$ such that $N \rightarrow \infty$. Recall that $\mathbf u(t)= (\mathbf u_1^\top(t), \dots, \mathbf u_n^\top(t) )^\top$, $\mathbf e(t)= ( e_1(t), \dots,  e_n(t) )^\top$, and $\mathbf x(t)= (\mathbf x_1^\top(t), \dots, \mathbf x_n^\top(t) )^\top$.
Consistent with the concept of competitive equilibrium in microeconomics \cite{mas1995microeconomic}, the infinite-horizon competitive equilibrium  can be defined as follows.

\medskip
\begin{definition}\label{definition2}
	Given a feasible initial condition $\mathbf x(0) \in \mathcal{X}_0$, the triplet $(\boldsymbol{\lambda}^\ast, \mathbf{U}^\ast, \mathbf{E}^\ast)$ is called an \emph{infinite-horizon competitive equilibrium}  if it meets the following conditions.
	\begin{itemize}
		\item[(i)] Under the equilibrium, each agent maximizes its finite payoff; i.e.,  $(\mathbf U_i^\ast, \mathbf E_i^\ast)$
		is an optimizer  to 
		\begin{equation}\label{opt_DLTD_5_3}
			\begin{aligned}
				\max_{\mathbf U_i,  \mathbf E_i} \quad &   \sum_{t=0}^{\infty} \Big[f_i(\mathbf x_i(t), \mathbf u_i(t))+ \lambda_t^\ast e_i(t) \Big] \\
				{\rm s.t.} \quad &\mathbf x_i(t+1) = \mathbf A_i \mathbf x_i(t)+ \mathbf B_i \mathbf u_i(t),\\
				& e_i(t)  \leq a_i(t) - h_i(\mathbf u_i(t)), \quad t \in \mathcal{T}.
			\end{aligned}
		\end{equation}
		\item[(ii)] Under the equilibrium, the traded resource is balanced at each time interval; i.e.,
		\begin{equation}\label{eq_balancing_2}
			\sum_{i=1}^n e_i^\ast(t) =0, \quad t \in \mathcal{T}.
      	\end{equation}
	\end{itemize}
\end{definition}

Similarly, social welfare can be maximized by finding an operating point $( \mathbf{U}^\ast, \mathbf{E}^\ast)$ which solves the following optimization problem 
\begin{equation}\label{opt_social_DLD_infinite_3}
	\begin{aligned}
		\max_{\mathbf{U}, \mathbf{E}} \quad &  \sum_{i=1}^{n} \sum_{t=0}^{\infty} f_i(\mathbf x_i(t), \mathbf u_i(t))\\
		{\rm s.t.} \quad & \mathbf x_i(t+1) = \mathbf A_i \mathbf x_i(t)+ \mathbf B_i \mathbf u_i(t), \\
		& e_i(t)  \leq a_i(t) - h_i(\mathbf u_i(t)), \\
		&\sum_{i=1}^n e_i(t) =0, 
		\quad   t \in \mathcal{T}, \,\,\, i \in \mathcal V.
	\end{aligned}
\end{equation}

In our previous work \cite{chen2021social}, we discussed the equivalence between a competitive equilibrium and a social welfare maximization solution over a finite horizon. Using duality theory, we proved that under some convexity assumptions, Slater's condition holds, leading to a zero duality gap between the two problems. For an infinite horizon case, however, Slater's condition is not well-established, giving rise to the complexity of the analysis. In  this section, we aim to examine the relation between a competitive equilibrium and a social welfare maximization solution over an infinite horizon, using the asymptotic behavior of the pricing sequence.
\subsection{MAS with General Cost Functions}\label{sec:existence}


\begin{theorem}\label{theorem6}
	Let $(\boldsymbol{\lambda}^\ast, \mathbf{U}^\ast, \mathbf{E}^\ast)$  be an infinite-horizon competitive equilibrium for a feasible initial condition $\mathbf x(0) \in \mathcal{X}_0$. Then $(\mathbf{U}^\ast, \mathbf{E}^\ast)$ maximizes social welfare over an infinite horizon. 
\end{theorem}
\begin{proof} 
	As the competitive equilibrium exists for $\mathbf x(0) \in \mathcal{X}_0$, the corresponding optimization problem in \eqref{opt_DLTD_5_3} is feasible. Additionally, the optimal price $\boldsymbol \lambda^\ast$ and the optimal solution $(\mathbf{U}^\ast, \mathbf{E}^\ast)$ are well-defined over the infinite horizon. So, we merge the optimization problems in \eqref{opt_DLTD_5_3} associated with all agents $i \in \mathcal{V}$ to obtain 
	\begin{equation}\label{opt_DLTD_5_5}
		\begin{aligned}
			\max_{\mathbf U,  \mathbf E} \quad &   \sum_{i=1}^{n}\sum_{t=0}^{\infty} f_i(\mathbf x_i(t), \mathbf u_i(t))+ \sum_{t=0}^{\infty} \lambda_t^\ast \sum_{i=1}^{n} e_i(t)  \\
			{\rm s.t.} \quad &\mathbf x_i(t+1) = \mathbf A_i \mathbf x_i(t)+ \mathbf B_i \mathbf u_i(t),\\
			& e_i(t)  \leq a_i(t) - h_i(\mathbf u_i(t)), \quad t \in \mathcal{T},  \,\,\, i \in \mathcal V,
		\end{aligned}
	\end{equation}
	such that $\sum_{i=1}^{n}e_i^\ast(t)=0$. Clearly, the solution $(\mathbf{U}^\ast, \mathbf{E}^\ast)$ which solves \eqref{opt_DLTD_5_5} is also a solution to \eqref{opt_social_DLD_infinite_3}. Additionally, both \eqref{opt_social_DLD_infinite_3} and \eqref{opt_DLTD_5_5} have the same optimal value which is finite. Hence,  social welfare maximization \eqref{opt_social_DLD_infinite_3}  is feasible for $\mathbf x(0) \in \mathcal{X}_0$ with the optimal solution $(\mathbf{U}^\ast, \mathbf{E}^\ast)$. 
\end{proof}

\medskip

\begin{assumption}\label{assumption4}
	(i) For  $i \in \mathcal{V}$, suppose  $f_i(\cdot)$ is a negative definite (ND) concave function;  (ii)  $h_i(\cdot)$ is a non-negative convex  function; (iii) both  $f_i(\cdot)$  and $h_i(\cdot)$ pass through the origin; (iv) there exists a lower bound $C>0$ on the  excess network supply such that $C(t) \geq C$ for $t \in \mathcal{T}$.
\end{assumption}

\medskip	

\begin{theorem}\label{theorem9}
	Let Assumption \ref{assumption4} hold. Suppose $(\mathbf A_i, \mathbf B_i)$ is controllable for $i \in \mathcal{V}$. Let $\mathcal{X}_0$ be the set of all feasible initial conditions for \eqref{opt_DLTD_5_3} and \eqref{opt_social_DLD_infinite_3}. Then there exists ${\mathds{X}}_{0} \subseteq \mathcal{X}_0$ such that
	\begin{itemize}
		\item[(i)] ${\mathds{X}}_{0}$ has  a nonempty interior containing the origin;
		\item[(ii)] for any $\mathbf x(0) \in \mathds{X}_0$, any social welfare maximization solution $(\mathbf{U}^\ast, \mathbf{E}^\ast)$ is part of a competitive equilibrium with $\boldsymbol \lambda^\ast =0$ over the infinite horizon.
	\end{itemize}
\end{theorem}
In the proof of this theorem, we use the  following lemma.

\medskip

\begin{lemma}\label{lemma5}
	Let the social welfare maximization problem be feasible for $\mathbf x(0) \in \mathcal{X}_0$. Then  \eqref{opt_social_DLD_infinite_3} can be solved through the following constrained optimal control problem
	\begin{equation}\label{opt_social_DLD_infinite_4}
		\begin{aligned}
			\max_{\mathbf{U}} \quad & \sum_{t=0}^{\infty}  \sum_{i=1}^{n}  f_i(\mathbf x_i(t), \mathbf u_i(t))\\
			{\rm s.t.} \quad & \mathbf x_i(t+1) = \mathbf A_i \mathbf x_i(t)+ \mathbf B_i \mathbf u_i(t), \\
			& \sum_{i=1}^{n} h_i(\mathbf u_i(t)) \leq C(t),
			\quad   t \in \mathcal{T}, \,\,\, i \in \mathcal V. 
		\end{aligned}
	\end{equation}
\end{lemma}
\begin{proof}
	Feasibility of  \eqref{opt_social_DLD_infinite_3} leads to the feasibility of \eqref{opt_social_DLD_infinite_4}. Suppose $\mathbf U^\ast$ is a solution to   \eqref{opt_social_DLD_infinite_4}. Decomposing $\mathbf U^\ast$ to $\mathbf u_i^\ast(t)$ for $t \in \mathcal{T}$ and $i \in \mathcal{V}$, we define 
	\begin{equation}\label{eq_e_3}
		e_i^\ast(t) := a_i(t) - h_i(\mathbf u_i^{\ast}(t))  + \frac{1}{n} \left(\sum_{i=1}^{n}h_i(\mathbf u_i^{\ast}(t))  - \sum_{i=1}^{n}a_i(t)\right), 
	\end{equation}
	which satisfies
	\begin{equation}\label{eq_e2_3}
		\begin{aligned}
			&e_i^\ast(t) \leq a_i(t) - h_i(\mathbf u_i^{\ast}(t)),\\
			& \sum_{i=1}^{n}  e_i^\ast(t)=0.
		\end{aligned}
	\end{equation}
	Constructing $ \mathbf E^\ast$ based on $e_i^\ast(t)$ defined in \eqref{eq_e_3}, we conclude  $(\mathbf U^\ast, \mathbf E^\ast)$ solves  \eqref{opt_social_DLD_infinite_3}.
\end{proof}

\noindent\textit{Proof of Theorem 7.}
According to  Lemma \ref{lemma5}, we examine  \eqref{opt_social_DLD_infinite_4}. Without loss of generality, replace the time-varying inequality constraint $\sum_{i=1}^{n} h_i(\mathbf u_i(t)) \leq C(t)$ in \eqref{opt_social_DLD_infinite_4} with the time-invariant constraint $\sum_{i=1}^{n} h_i(\mathbf u_i(t)) <C$,  which represents a constraint set $\Omega \subset \mathbb{R}^{n m}$ on  control inputs, containing the origin in its interior. 
For $i \in \mathcal{V}$, consider all controllers restrained by $\| \mathbf u_i(t)\|_{\infty} \leq \epsilon$ for sufficiently small $\epsilon>0$ such that the controllers  lie in the interior of the constraint set $\Omega$. According to \cite{hu2002null}, the set of all feasible initial conditions that can be steered to the origin in a finite time  when  control inputs satisfy   saturation constraints $\| \mathbf u_i(t)\|_{\infty} \leq \epsilon$ has a non-empty interior containing the origin; to be consistent with the literature, we denote this set of initial conditions as the null controllable region $\mathcal{C}$. Since $f_i(\cdot)$ is a negative definite function passing through the origin, the problem in \eqref{opt_social_DLD_infinite_4} is feasible for initial conditions $\mathbf x(0) \in \mathcal{C}$. Now select a neighborhood of the origin  $\mathds{X}_0 \subseteq \mathcal{C}$ such that all $\mathbf x(0) \in \mathds{X}_0$ satisfy $\| \mathbf x(0)\| \leq \bar\epsilon$. If $\bar \epsilon$ is selected sufficiently small then the optimal control $\mathbf u_i^\ast(t)$ which solves \eqref{opt_social_DLD_infinite_4} lies in the set $\Omega$, i.e., $\sum_{i=1}^{n} h_i(\mathbf u_i^\ast(t)) < C$ for $t \in \mathcal{T}$.  This is because $h_i(\cdot)$ is  convex  in an open neighborhood of the origin, and therefore, continuous \cite{peajcariaac1992convex}. Hence, the optimal control input for \eqref{opt_social_DLD_infinite_4} would be the solution to the following unconstrained optimal control problem 
\begin{equation}\label{opt_social_DLD_infinite_5}
	\begin{aligned}
		\max_{\mathbf{U}} \quad & \sum_{t=0}^{\infty}  \sum_{i=1}^{n}  f_i(\mathbf x_i(t), \mathbf u_i(t))\\
		{\rm s.t.} \quad & \mathbf x_i(t+1) = \mathbf A_i \mathbf x_i(t)+ \mathbf B_i \mathbf u_i(t),
		\,\,\,   t \in \mathcal{T}, \,\, i \in \mathcal V. 
	\end{aligned}
\end{equation}
For more insights into the sets, see Fig. \ref{fig_sets}.
\begin{figure}[!t]
	\centering
	\includegraphics[width=1.75 in]{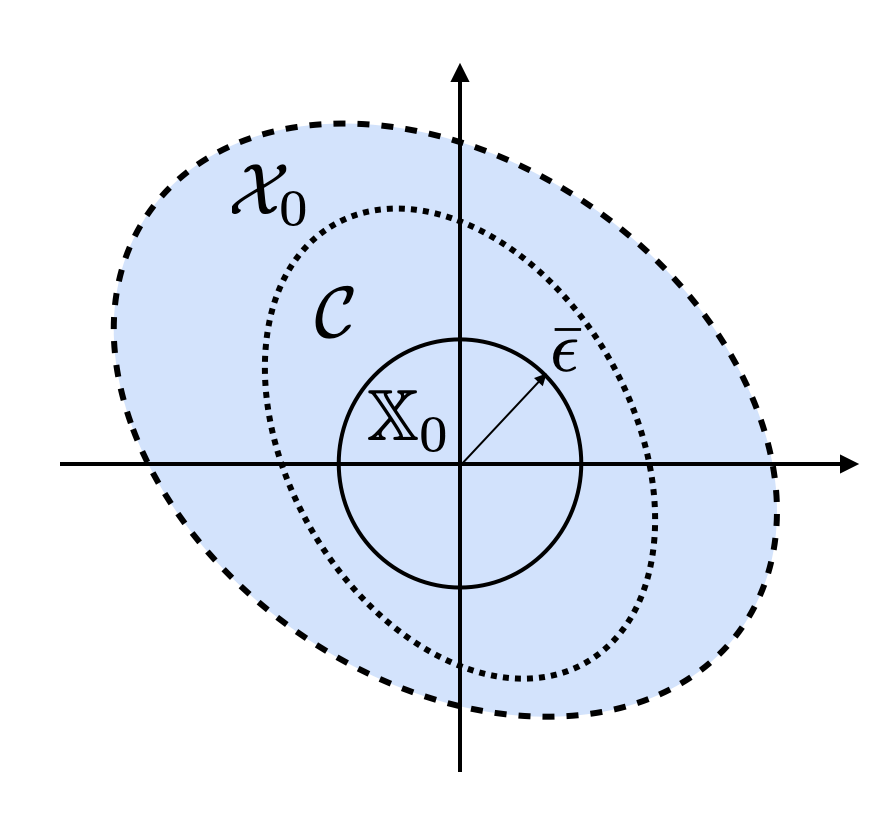}
	\caption{The three feasible sets in the proof of Theorem \ref{theorem9}.}
	\label{fig_sets}
\end{figure}
The problem in \eqref{opt_social_DLD_infinite_5} is separable, so $\mathbf U^\ast$ is also a solution to 
\begin{equation}\label{opt_social_DLD_infinite_6}
	\begin{aligned}
		\max_{\mathbf{U}_i} \quad & \sum_{t=0}^{\infty}   f_i(\mathbf x_i(t), \mathbf u_i(t))\\
		{\rm s.t.} \quad & \mathbf x_i(t+1) = \mathbf A_i \mathbf x_i(t)+ \mathbf B_i \mathbf u_i(t),
		\quad   t \in \mathcal{T},
	\end{aligned}
\end{equation}
for $i \in \mathcal{V}$. Define $e_i^\ast(t)$ as in \eqref{eq_e_3} which leads to \eqref{eq_e2_3}. 
Denoting $\mathbf E_i^\ast = (e_i^\ast(0), \cdots, e_i^\ast(\infty))^\top$ for $i \in \mathcal{V}$, we conclude $(\mathbf U_i^\ast, \mathbf E_i^\ast)$ is an optimizer for the optimization problem in \eqref{opt_DLTD_5_3} under $\lambda^\ast_t=0$ such that \eqref{eq_balancing_2} is satisfied. Consequently, $(\mathbf U^\ast, \mathbf E^\ast)$, along with $\lambda^\ast_t=0$ for $t \in \mathcal{T}$, constitutes a competitive equilibrium.
{\hfill$\square$}

\medskip

\begin{theorem}\label{theorem8}
	Let Assumption \ref{assumption4}  hold. Suppose $\mathbf x(0) \in \mathcal{X}_0$. Assume $(\boldsymbol \lambda^\ast, \mathbf{U}^\ast, \mathbf{E}^\ast)$ is a competitive equilibrium associated with $\mathbf x(0)$. Then there exists a time step $\bar N \left(\mathbf x(0)\right)$ such that $\lambda^\ast_t =0$ for $t > \bar N \left(\mathbf x(0)\right)$. 
\end{theorem}
\begin{proof}
	As the competitive equilibrium exists for $\mathbf x(0) \in \mathcal{X}_0$, the optimal price $\boldsymbol \lambda^\ast$ and the optimal solution $(\mathbf{U}^\ast, \mathbf{E}^\ast)$ are well-defined over the infinite horizon. According to Theorem \ref{theorem6},  $(\mathbf{U}^\ast, \mathbf{E}^\ast)$ maximizes  social welfare. Considering Lemma \ref{lemma5}, we  examine \eqref{opt_social_DLD_infinite_4}. Since the problem is feasible,  the optimal value is finite. Considering that $f_i(\mathbf x_i(t), \mathbf u_i(t))$ is a negative definite function, we obtain $\lim_{t \rightarrow \infty} f_i(\mathbf x_i(t), \mathbf u_i(t)) =0$. Concavity of $f_i(\cdot)$  in an open neighborhood of the origin implies its continuity. As $f_i(\cdot)$ is a negative definite concave function passing through the origin, we obtain $\lim_{t \rightarrow \infty} \mathbf x_i(t) =0$ and $\lim_{t \rightarrow \infty} \mathbf u_i^\ast(t) =0$ for $i \in \mathcal{V}$. Similarly, $h_i(\cdot)$  is continuous in an open  neighborhood of the origin.   Therefore,  there exists a finite time $\bar N \left(\mathbf x(0)\right)$ such that $\sum_{i=1}^{n} h_i(\mathbf u_i^\ast(t)) < C$ for $t > \bar N\left(\mathbf x(0)\right)$. In the remainder of this proof, by $\bar N$ we mean $\bar N\left(\mathbf x(0)\right)$. Decompose $\mathbf U^\ast$  to $\mathbf u^\ast(t)$ for $t \in \mathcal{T}$. Applying the first $\bar N+1$ optimal control inputs $\mathbf u^\ast(0), \dots, \mathbf u^\ast(\bar N)$ to the system dynamic, we reach $\mathbf x(\bar N+1)$. Denote $\mathbf U^\ast_{\bar N+1} = (\mathbf u^{\ast\top}(\bar N+1), \dots, \mathbf u^{\ast\top}(\infty))^\top$. Based on the principle of optimality, $\mathbf{U}^\ast_{\bar N+1}$  solves  the unconstrained optimal control problem in \eqref{opt_social_DLD_infinite_5} starting from $t=\bar N+1$ with the initial condition $\mathbf x(\bar N+1)$. Also, $\mathbf{U}^\ast_{\bar N+1}$  solves \eqref{opt_social_DLD_infinite_6} starting from $t=\bar N+1$ with the initial condition $\mathbf x(\bar N+1)$ for $i \in \mathcal{V}$. According to the proof of Theorem \ref{theorem9}, the solution to \eqref{opt_social_DLD_infinite_6} constitutes a competitive equilibrium with zero price. Therefore,   $\lambda_t^\ast=0$, for $t > \bar N$.
\end{proof}
\subsection{MAS with Quadratic Cost Functions}\label{sec:quadratic}
As a benchmark, we examine a quadratic case and provide some analytical results accordingly. 

Considering $f_i(\cdot)=f(\cdot; \theta_i)$ and $h_i(\cdot)$ as in Assumption \ref{assumption2}, the optimization problem  \eqref{opt_DLTD_5_3} can be rearranged as
\begin{equation}\label{opt_DLTD_5}
	\begin{aligned}
		\max_{\mathbf U_i,  \mathbf E_i} \quad &   \sum_{t=0}^{\infty} \Big[- \mathbf x_i^\top(t) \mathbf Q_i \mathbf x_i(t)  - \mathbf u_i^\top(t) \mathbf R_i \mathbf u_i(t)+ \lambda_t^\ast e_i(t) \Big] \\
		{\rm s.t.} \quad &\mathbf x_i(t+1) = \mathbf A_i \mathbf x_i(t)+ \mathbf B_i \mathbf u_i(t),\\
		& e_i(t)  \leq a_i(t) - \mathbf u_i^{\top}(t) \mathbf H_i \mathbf u_i(t), \quad t \in \mathcal{T},
	\end{aligned}
\end{equation}
with the balancing equality constraint $\sum_{i=1}^n e_i^\ast(t) =0$.
Also, the social welfare maximization problem  \eqref{opt_social_DLD_infinite_3} becomes
\begin{equation}\label{opt_social_DLD_infinite}
	\begin{aligned}
		\max_{\mathbf{U}, \mathbf{E}} \quad &  \sum_{i=1}^{n} \sum_{t=0}^{\infty} \Big[ - \mathbf x_i^\top(t) \mathbf Q_i \mathbf x_i(t)  - \mathbf u_i^\top(t) \mathbf R_i \mathbf u_i(t) \Big]\\
		{\rm s.t.} \quad & \mathbf x_i(t+1) = \mathbf A_i \mathbf x_i(t)+ \mathbf B_i \mathbf u_i(t), \\
		& e_i(t)  \leq a_i(t) - \mathbf u_i^{\top}(t) \mathbf H_i \mathbf u_i(t), \\
		&\sum_{i=1}^ne_i(t) =0, 
		\quad   t \in \mathcal{T}, \,\,\, i \in \mathcal V. 
	\end{aligned}
\end{equation}

Recall that $\mathbf x(t)$ and $\mathbf u(t)$ are the vectors incorporating the  states and  control inputs of all agents at time step $t \in \mathcal{T}$, respectively. Denote $\mathbf A = \text{blockdiag} \{\mathbf A_1, \dots, \mathbf A_n\}$, $\mathbf B = \text{blockdiag} \{\mathbf B_1, \dots, \mathbf B_n\}$, and $\mathbf H = \text{blockdiag} \{\mathbf H_1, \dots, \mathbf H_n\}$.  Additionally, denote $\mathbf Q = \text{blockdiag} \{\mathbf Q_1, \dots, \mathbf Q_n\}$ and
$\mathbf R = \text{blockdiag} \{\mathbf R_1, \dots, \mathbf R_n\}$. According to Lemma \ref{lemma5}, instead of  social welfare maximization  \eqref{opt_social_DLD_infinite}, we can examine 
\begin{equation}\label{opt_social_infinite_2}
	\begin{aligned}
		\max_{\mathbf{U}} \quad &   \sum_{t=0}^{\infty} \Big[ - \mathbf x^\top(t) \mathbf Q \mathbf x(t)  - \mathbf u^\top(t) \mathbf R \mathbf u(t) \Big]\\
		{\rm s.t.} \quad & \mathbf x(t+1) = \mathbf A \mathbf x(t)+ \mathbf B \mathbf u(t), \\
		& \mathbf u^{\top}(t) \mathbf H \mathbf u(t)  \leq C(t),
		\quad   t \in \mathcal{T},
	\end{aligned}
\end{equation}
which is a CLQR problem. Suppose $\mathbf P \in \mathbb{R}^{n d \times n d}$ is a solution to the  discrete algebraic Riccati equation 
$
	\mathbf P= \mathbf A^\top \mathbf P \mathbf A + \mathbf Q - \mathbf A^\top \mathbf P \mathbf B (\mathbf B^\top \mathbf P \mathbf B + \mathbf R)^{-1} \mathbf B^\top \mathbf P \mathbf A
$.
Define $\mathbf K  = -(\mathbf R+ \mathbf B^\top \mathbf P \mathbf B)^{-1} \mathbf B^\top \mathbf P \mathbf A$. 
Select
\begin{equation}\label{set_ROA}
	\resizebox {0.485\textwidth} {!} {$
		{\mathds{X}}_{0} = \left\{ \mathbf x(0) \in \mathbb{R}^{n d}: \left\| \mathbf x(0) \right\|^2 \leq \frac{C\sigma_{\rm min}(\mathbf P)}{\sigma_{\rm max}(\mathbf P)\sigma_{\rm max}(\mathbf K^\top \mathbf H \mathbf K)} \right\}. $}
\end{equation}


\medskip

\begin{theorem}\label{theorem5}
	Suppose $f_i(\cdot) = f(\cdot; \theta_i)$ and $h_i(\cdot)$ are as in Assumption \ref{assumption2}. Let $(\mathbf A, \mathbf B)$ be controllable, and $C(t)\geq C>0$ for $t \in \mathcal{T}$. 
	If  $\mathbf x(0) \in \mathds{X}_0$ in \eqref{set_ROA}, then  the following statements hold.
	\begin{itemize}
		\item[(i)] The social welfare maximization problem   \eqref{opt_social_DLD_infinite} is feasible with the following optimal control input 
		\begin{equation}\label{eq_u}
			\resizebox {0.446\textwidth} {!} {$
				\mathbf u^\ast(t)= \mathbf K \mathbf x(t) = -(\mathbf R+ \mathbf B^\top \mathbf P \mathbf B)^{-1} \mathbf B^\top \mathbf P \mathbf A \mathbf x(t), \, t \in \mathcal{T}. $}
		\end{equation}
		\item [(ii)] The optimal price associated with the competitive equilibrium  is $\lambda^\ast_t = 0$ for $t \in \mathcal{T}$.
	\end{itemize}
	
\end{theorem}
\begin{proof}
	(i) Consider $\beta>0$. Any sublevel set 
	\begin{equation}\label{set_invariant}
		\mathds{X}=\{ \mathbf x \in \mathbb{R}^{n d} : \mathbf x^\top \mathbf P \mathbf x \leq \beta\}
	\end{equation}
	is forward invariant for the closed-loop system $\mathbf x(t+1) = (\mathbf A + \mathbf B \mathbf K)\mathbf x(t)$ \cite{Borrelli}. If $\beta$ is selected such that the control input $\mathbf u(t) = \mathbf K \mathbf x(t)$ satisfies the inequality constraint in \eqref{opt_social_infinite_2} for all $\mathbf x(t) \in \mathds{X}$, then the origin is locally asymptotically stable 
	and the region of attraction contains $\mathds{X}$ for the CLQR problem \eqref{opt_social_infinite_2}.
	Considering $\mathbf u(0) = \mathbf K \mathbf x(0)$, we obtain $\mathbf u^{\top}(0) \mathbf H \mathbf u(0) = \mathbf x^\top(0) \mathbf K^\top  \mathbf H \mathbf K \mathbf x(0)$.
	Additionally, the following inequality holds
	\begin{multline}\label{eq57}
		\sigma_{\rm min}(\mathbf K^\top  \mathbf H \mathbf K) \left\| \mathbf x(0) \right\|^2 \leq \mathbf x^\top(0) \mathbf K^\top  \mathbf H \mathbf K \mathbf x(0)   \\ \leq \sigma_{\rm max}(\mathbf K^\top  \mathbf H \mathbf K) \left\| \mathbf x(0) \right\|^2.
	\end{multline}
	According to \eqref{set_ROA}, we have $\sigma_{\rm max}(\mathbf K^\top \mathbf H \mathbf K) \left\| \mathbf x(0) \right\|^2 \leq C\sigma_{\rm min}(\mathbf P)/\sigma_{\rm max}(\mathbf P)$. Consequently, \eqref{eq57} yields
	\begin{equation*}
		\mathbf u^{\top}(0) \mathbf H \mathbf u(0) \leq  \frac{C\sigma_{\rm min}(\mathbf P)}{\sigma_{\rm max}(\mathbf P)} \leq C,
	\end{equation*}
	which satisfies the inequality constraint in \eqref{opt_social_infinite_2}. Additionally, the set $\mathds{X}_0$ in \eqref{set_ROA} is a subset of the forward invariant set in \eqref{set_invariant} with the choice of $\beta = {C \sigma_{\rm min}(\mathbf P)}/{\sigma_{\rm max}(\mathbf K^\top \mathbf H \mathbf K)}$. 
	Since $\mathbf u(0) = \mathbf K \mathbf x(0)$ satisfies the inequality constraint and $ \mathbf x(0) $ lies in the invariant set $\mathds X$ in \eqref{set_invariant},  then $\mathbf x(1)$ also lies in the invariant set $\mathds X$ in \eqref{set_invariant}; i.e.,
	$
	\mathbf x^\top(1) \mathbf P \mathbf x(1) \leq {C \sigma_{\rm min}(\mathbf P)} / {\sigma_{\rm max}(\mathbf K^\top \mathbf H \mathbf K)}.
	$
	Considering $\sigma_{\rm min}(\mathbf P) \left\| \mathbf x(1) \right\|^2 \leq \mathbf x^\top(1) \mathbf P \mathbf x(1)$, we yield $\left\| \mathbf x(1) \right\|^2 \leq {C }/{\sigma_{\rm max}(\mathbf K^\top \mathbf H \mathbf K)}$.
	Then, the feedback law $\mathbf u(1) = \mathbf K \mathbf x(1)$ leads to
$
		\mathbf u^{\top}(1) \mathbf H \mathbf u(1) \leq \sigma_{\rm max}(\mathbf K^\top  \mathbf H \mathbf K) \left\| \mathbf x(1) \right\|^2 \leq C
$,
	which satisfies the inequality constraint. Using the same logic, all future states $\mathbf x(t)$ lie in the invariant set $\mathds X$ in \eqref{set_invariant}, and all future control inputs $\mathbf u(t)= \mathbf K \mathbf x(t)$ satisfy the inequality constraint. So, with the initial conditions in \eqref{set_ROA}, we treat \eqref{opt_social_infinite_2} as an unconstrained LQR problem. Since $(\mathbf A, \mathbf B)$ is controllable, $\mathbf Q>0$, and $\mathbf R>0$,  the optimization problem \eqref{opt_social_infinite_2} is feasible with the optimal solution \eqref{eq_u}.
	According to Lemma \ref{lemma5},  the social welfare maximization problem  \eqref{opt_social_DLD_infinite} is also feasible and the optimal control is as \eqref{eq_u}. 
	
	(ii) Follows directly from  Theorem \ref{theorem9}.
\end{proof}


\subsection{Tracking Problem}
In standard regulation problems, it is desirable to steer the state to zero. However, in  real-world applications, the state often follows a reference $\bar {\mathbf x}$, which is a well-known class of tracking problems. We can convert a tracking problem into a standard regulation problem by  change of variables $\mathbf x_{i,e}(t):= \mathbf x_{i}(t) - \bar{\mathbf x}$ and  $\mathbf u_{i,e}(t):=\mathbf u_i(t)-\bar{\mathbf u}_i$, where $\bar{\mathbf x}$ is the desired state and $\bar{\mathbf u}_i$ is the steady state control input (see Example 10.3 in \cite{tedrake2009underactuated}). In this section, we examine extensions to our proposed results in the context of tracking problems.

Consider the following modified optimization problem instead of \eqref{opt_DLTD_5} to obtain a competitive equilibrium
\begin{equation}\label{eq59}
		\resizebox {0.485\textwidth} {!} {$
	\begin{aligned}
		\max_{\mathbf U_{i}^e, \mathbf E_i} \quad &   \sum_{t=0}^{\infty} \Big[- \mathbf x_{i,e}^\top(t) \mathbf Q_i \mathbf x_{i,e}(t)  - \mathbf u_{i,e}^\top(t) \mathbf R_i \mathbf u_{i,e}(t)+ \lambda_t^\ast e_i(t) \Big] \\
		{\rm s.t.} \quad & \mathbf x_{i,e}(t+1) = \mathbf A_i \mathbf x_{i,e}(t)+ \mathbf B_i \mathbf u_{i,e}(t), \\
		& e_i(t)  \leq a_i(t) - (\mathbf u_{i,e}(t)+\mathbf {\bar{u}}_i)^{\top} \mathbf H_i (\mathbf u_{i,e}(t)+\mathbf {\bar{u}}_i), 
		\, \, \,   t \in \mathcal{T},
	\end{aligned} $}
\end{equation}
where $\mathbf U_i^e=(\mathbf u_{i,e}^\top(0), \dots, \mathbf u_{i,e}^\top(\infty))^\top$.
The inequality constraint in \eqref{eq59} can be rearranged as $e_i(t)  \leq a_{i,e}(t) - h_{i,e}(\mathbf u_{i,e}(t))$, where
$a_{i,e}(t):= a_i(t) - \bar{\mathbf u}_i^\top \mathbf H_i  \bar{\mathbf u}_i$ and $h_{i,e}(\mathbf u_{i,e}(t)):= \mathbf u_{i,e}^\top(t) \mathbf H_i \mathbf u_{i,e}(t) +2  \bar{\mathbf u}_i^\top \mathbf H_i \mathbf u_{i,e}(t)$.  With this change of variables, \eqref{eq59} can be represented as the standard form  \eqref{opt_DLTD_1}.

To be consistent with the previous assumptions, let $\sum_{i=1}^{n} a_{i,e}(t)>0$, where  $\sum_{i=1}^{n} a_{i}(t)>\sum_{i=1}^{n} \bar{\mathbf u}_i^\top \mathbf H_i  \bar{\mathbf u}_i$. Note that $h_{i,e}(\mathbf u_{i,e}(t))$ is not a non-negative function, which causes no issue with the presented results. The non-negativity of $h_i(\cdot)$ is mainly related to its physical interpretations that  represent energy.
The social welfare maximization problem can be defined in a similar manner.

The community microgrid, introduced in Section \ref{sec:app_microgrid}, can be formulated as a tracking problem in which $\bar{\mathbf x}$ represents the desired temperature  ($\tccentigrade$).

\section{Illustrative Examples}\label{sec:examples}
In this section, we simulate some standard regulation problems to validate the proposed theorems.
\subsection{Social Shaping over a Finite Horizon}\label{sec:example_finite}
Consider a quadratic MAS consisting of $3$ agents  in the time horizon $N=6$, with the following state-space matrices
\begin{equation*}
		\small
	\begin{aligned}
		&{\mathbf A_1} = \left[ {\begin{array}{*{20}{c}}
				{ 0.4}&{ - 0.1}&{0.2}\\
				{0.2}&{  0.3}&{0.1}\\
				{0.3}&{ - 0.1}&{-0.2}
		\end{array}} \right],	\, {\mathbf B_1} = \left[ {\begin{array}{*{20}{c}}
				4&5\\
				2&1\\
				3&5
		\end{array}} \right], \\ &{\mathbf A_2} = \left[ {\begin{array}{*{20}{c}}
				{-0.1}&{0.2}&{ -0.3}\\
				{0.3}&{  0.4}&{ - 0.1}\\
				{ -0.1}&{0.2}&{  -0.7}
		\end{array}} \right], \, {\mathbf B_2} = \left[ {\begin{array}{*{20}{c}}
				1&4\\
				2&5\\
				6&3
		\end{array}} \right], \\ &{\mathbf A_3} = \left[ {\begin{array}{*{20}{c}}
				{  0.5}&{-0.2}&{0.6}\\
				{ -0.4}&{0.9}&{0.3}\\
				{0.5}&{0.3}&{ - 0.8}
		\end{array}} \right], \,
		{\mathbf B_3} = \left[ {\begin{array}{*{20}{c}}
				2&3\\
				1&2\\
				5&4
		\end{array}} \right],
	\end{aligned}
\end{equation*}
and the initial states ${\mathbf x_1}(0)=(25, 35, 75)^\top$, ${\mathbf x_2}(0)=(40, 50, 70)^\top$, and ${\mathbf x_3}(0)=(50, 80, 90)^\top$. 
Additionally, suppose agents have the representative excess  resources $a_1(t)=-\sin(\frac{\pi}{6}t)+1.2$, $a_2(t)=-2\sin(\frac{\pi}{6}t)+2.2$, and $a_3(t)=0$. The total excess network generation is obtained as $C(t)=\sum_{i=1}^{3}a_i(t)=-3\sin(\frac{\pi}{6}t)+3.4$. 
Considering Assumption \ref{assumption2}, agents have $\mathbf R_1 = \mathbf R_2 = \mathbf R_3 = 0.3 \mathbf I$, and
\begin{equation*}
	\begin{split}
		{\mathbf H_1} = \left[ {\begin{array}{*{20}{c}}
				2&3\\
				3&6
		\end{array}} \right], \, {\mathbf H_2} = \left[ {\begin{array}{*{20}{c}}
				1&-2\\
				-2&5
		\end{array}} \right], \, {\mathbf H_3} = \left[ {\begin{array}{*{20}{c}}
				4&{1}\\
				{1}&3
		\end{array}} \right].
	\end{split}
\end{equation*}
Let $\lambda^\dag =20$. We aim to design $\mathbf Q_i$ for $i \in \mathcal{V}$ such that $\lambda^\ast_t \leq \lambda^\dag$ at all time steps. 
\medskip
\subsubsection{Analytical Approach}  
 Using Theorems \ref{theorem2} and \ref{theorem3}, the upper bound $\delta_{\rm max}$ for the personalized parameter $\mathbf Q_i$ is obtained as $0.00017$ and $0.00018$, receptively. Theorem \ref{theorem3} provides a larger upper bound compared to Theorem \ref{theorem2}. Therefore, we assign $\delta_{\rm max}=0.00018$ and select $\mathbf Q_i=0.00018 \mathbf I$ for $i \in \mathcal{V}$. We solve  the social welfare maximization in \eqref{opt_social_DLD_1} to obtain   $-\lambda^\ast_t$ as  the Lagrange multiplier associated with the equality constraint $\sum_{i=1}^{n} e_i(t)=0$ for $t \in \mathcal{T}$. The optimal prices $\lambda^\ast_t$  are depicted in Fig. \ref{fig_2}.  
\begin{figure}[t]
	\centering
	\begin{subfigure}{\columnwidth}
		\centering
		\includegraphics[width=.6\columnwidth]{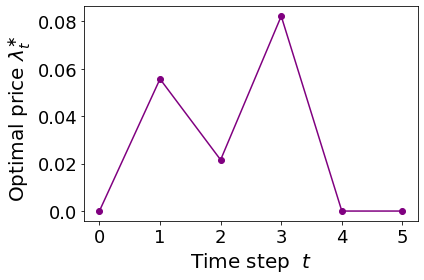}
		\caption{ $\delta_{\rm max}=0.00018$, obtained from Theorem \ref{theorem3}.}
		\label{fig_2}
	\end{subfigure}%
	\quad
	\hfill
	\begin{subfigure}{\columnwidth}
		\centering
		\includegraphics[width=.6\columnwidth]{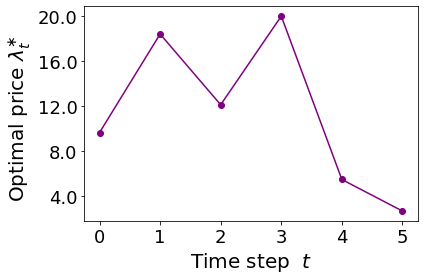}
		\caption{$\delta_{\rm max}=0.024$, obtained from Algorithm \ref{algorithm1}.}
		\label{fig_3}
	\end{subfigure}
	\caption{The optimal price $\lambda^\ast_t$ over time steps.}
\end{figure}
As can be seen, the maximum value of the optimal price is $0.08$ which is much less than the price threshold $\lambda^\dag=20$, and therefore, socially resilient. This confirms that Theorems \ref{theorem2} and\ref{theorem3} are valid although they provide conservative results.

\medskip
\subsubsection{Numerical approach} We run Algorithm \ref{algorithm1}  for $30$ steps with the choice of $d_\varrho = 1$ which is sufficiently large and satisfies $\bar \lambda^\ast(d_\varrho)=835.9> \lambda^\dag$. The algorithm converges to $\delta_{\rm max}=0.024$ which validates Theorem \ref{theorem7}. Selecting $\mathbf Q_i = 0.024 \mathbf I$ for $i \in \mathcal{V}$, the optimal prices are obtained in Fig. \ref{fig_3}, which are less than or equal to $20$. The maximum value of the price over the entire horizon is $\lambda^\ast_3 = 20$, occurring at time step $t=3$. This validates Theorem \ref{theorem4} and shows that the numerical algorithm provides a tight upper bound for the personalized parameter $\mathbf Q_i$, and thus works well in practice.

\subsection{Competitive Equilibrium over an Infinite Horizon}
Consider a dynamic MAS with $3$ agents with the following state matrices
\begin{equation*}
	\resizebox {0.48\textwidth} {!} {$
	\begin{aligned}
		&{\mathbf A_1} = \left[ {\begin{array}{*{20}{c}}
				{ 1.1}&{ - 0.5}&{1.8}\\
				{-0.4}&{  0.6}&{0.7}\\
				{-0.3}&{  0.7}&{-0.6}
		\end{array}} \right],  {\mathbf A_2} = \left[ {\begin{array}{*{20}{c}}
				{0.4}&{1.2}&{ -0.1}\\
				{-0.8}&{  -1.3}&{  0.6}\\
				{ 0.1}&{0.7}&{  0.5}
		\end{array}} \right], \\
	&{\mathbf A_3} = \left[ {\begin{array}{*{20}{c}}
				{  0.6}&{-1.2}&{0.9}\\
				{ -1.4}&{0.7}&{0.3}\\
				{-1.5}&{0.7}&{  0.1}
		\end{array}} \right].
	\end{aligned}
$}
\end{equation*}
Suppose $\mathbf B_i$, $\mathbf R_i$, and $\mathbf H_i$ are  selected according to Section \ref{sec:example_finite} and $\mathbf Q_i=0.005 \mathbf I$ for $i \in \mathcal{V}$. Suppose agents provide excess supply resources  $a_1(t)=1$, $a_2(t)=1.8$, and $a_3(t)=0$, which leads to  the total excess network supply $C(t)=\sum_{i=1}^{3}a_i(t)=2.8$. The forward invariant set $\mathds{X}_0$ in \eqref{set_ROA} represents an $L2$-ball of radius $0.47$ centered at the origin. In the following, we study two scenarios on the choice of initial conditions.

\medskip
\subsubsection{Scenario I (Initial Conditions Outside the Invariant Set $\mathds{X}_0$)}
Suppose agents have initial conditions as in \ref{sec:example_finite} which lie outside the set $\mathds {X}_0$ in \eqref{set_ROA}. Let  $\mathbf U^\ast_{C,ij}$ and $\mathbf U^\ast_{S,ij}$  denote the $j$-th optimal inputs of agent $i$ which solve the competitive problem \eqref{opt_DLTD_5} and the social welfare maximization \eqref{opt_social_DLD_infinite},  respectively.  Similarly,  denote $\mathbf E^\ast_{C,i}$  and $\mathbf E^\ast_{S,i}$  as the optimal trading decisions of agent $i$ associated with the competitive equilibrium and the social welfare maximization solution, respectively. Similar to the finite horizon case in \cite{chen2021social}, we solve  social welfare maximization  \eqref{opt_social_DLD_infinite} to obtain $-  \lambda^\ast_t$ as the Lagrange multiplier of the equality constraint $\sum_{i=1}^{n}e_i(t)=0$. Additionally, we acquire the associated optimal solution $(\mathbf U^\ast_{S,ij}, \mathbf E^\ast_{S,i})$ where $i \in \{1,2,3\}$ and $j \in \{1,2\}$, as  depicted in Fig. \ref{fig_5}. Using the obtained optimal price $\lambda_t^\ast$, we solve the optimization problem  \eqref{opt_DLTD_5}  to reach the competitive equilibrium $(\mathbf U^\ast_{C,ij}, \mathbf E^\ast_{C,i})$ for $i \in \{1,2,3\}$ and $j \in \{1,2\}$, as  illustrated in Fig. \ref{fig_5}. According to Fig. \ref{fig_5}, the competitive equilibrium and the solution of the social welfare maximization coincide,  validating Theorem \ref{theorem6}. 
\begin{figure}[!t]
	\centering
	\includegraphics[width=8.4cm]{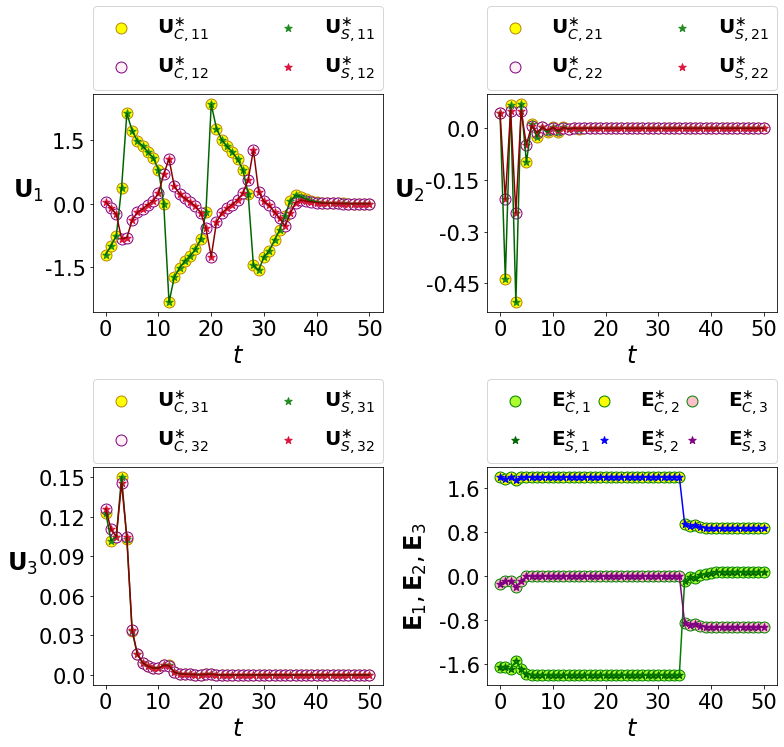}
	\caption{The competitive equilibrium and the social welfare maximization solution, in Scenario I.}
	\label{fig_5}
\end{figure}
Additionally, the values of  optimal prices $\lambda^\ast_t$ over time intervals are depicted in Fig. \ref{fig_6}, which indicates  $\lambda^\ast_t=0$ for $t >35$. This  validates Theorem \ref{theorem8}.
\begin{figure}[!t]
	\centering
	\includegraphics[width=2.5 in]{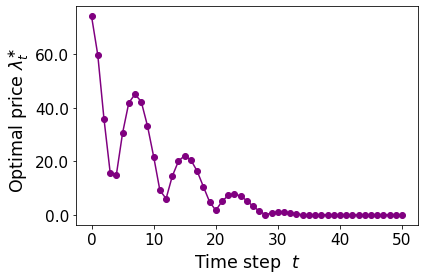}
	\caption{The optimal price over time intervals, in Scenario I.}
	\label{fig_6}
\end{figure}

\medskip
\subsubsection{Scenario II (Initial Conditions Inside the Invariant Set $\mathds{X}_0$)}
 Suppose agents have  initial conditions ${\mathbf x_1}(0)=(0.2,0.1,0.08)^\top$, ${\mathbf x_2}(0)=(0.1,0.06,0.3)^\top$, and ${\mathbf x_3}(0)=(0.5,0.2,0.1)^\top$, 
that satisfy  $\mathbf x(0) \in \mathds{X}_0$ in \eqref{set_ROA}. Using the same procedure as in Scenario I, we obtain the optimal control input $\mathbf U^\ast_{S,ij}$ for $i \in \{1,2,3\}$ and $j \in \{1,2\}$. Additionally,  we compute the linear feedback law in \eqref{eq_u} and denote it as $\mathbf U^\ast_{K,ij}$, representing the $j$-th control input of agent $i$, for $i \in \{1,2,3\}$ and $j \in \{1,2\}$.  Both $\mathbf U^\ast_{S,ij}$  and $\mathbf U^\ast_{K,ij}$  are depicted in Fig. \ref{fig_10}, which illustrates the two control policies are the same. In Fig. \ref{fig_10}, we present optimal prices over respective time intervals, indicating  $\lambda_t^\ast=0$ at all time intervals. Our results validate Theorems \ref{theorem9} and \ref{theorem5}.
\begin{figure}[!t]
	\centering
	\includegraphics[width=8.4cm]{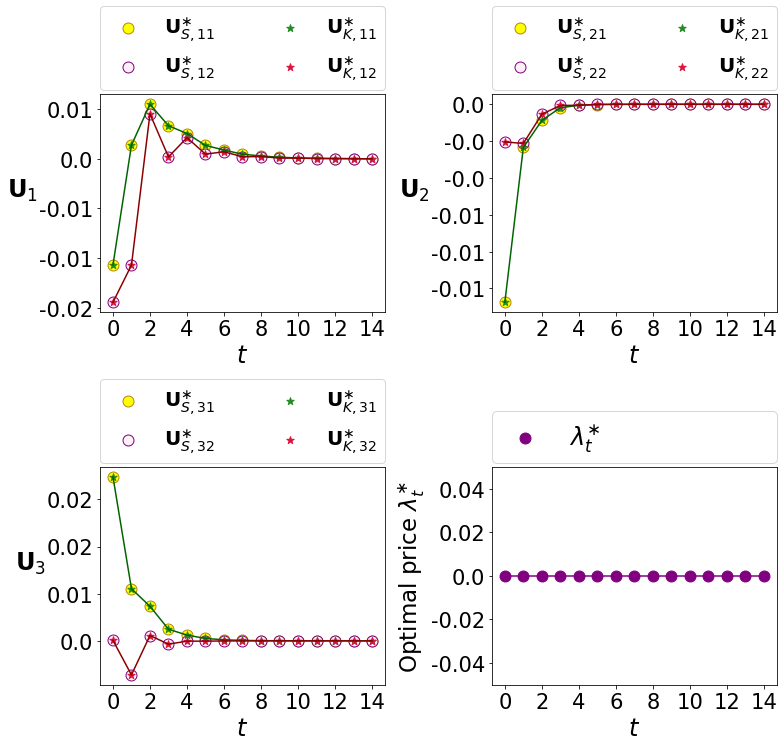}
	\caption{The optimal control inputs and the optimal price $\lambda_t^\ast$ over time intervals, in Scenario II.}
	\label{fig_10}
\end{figure}

\section{Conclusion}\label{sec:conclusions}
In this paper, we investigated the properties of a  competitive equilibrium, in particular the price trajectory, in a self-sustained dynamic MAS.  First, we focused on a finite horizon case and showed the social shaping problem is solvable both implicitly (for general classes of utility functions) and explicitly (for quadratic MAS). Furthermore, we presented a numerical algorithm which provides more accurate results  compared to the proposed analytical solutions. Secondly, we studied  a competitive equilibrium over an infinite horizon. We examined the relationship between a competitive equilibrium and the solution of the social welfare maximization problem. Additionally, we investigated the decaying behavior of the price trajectory which depends on the system's initial state.  Finally, we  focused on  quadratic MAS and the associated CLQR problem to obtain explicit results validating the previous findings. As future work, it might be possible  to extend the results to systems with nonlinear dynamics in the presence of disturbances and uncertainties. Additionally, one might consider adding more physical constraints to the framework. 

\section*{Acknowledgment}

This work was supported by the Australian Research Council under grants DP190102158, DP190103615, and LP210200473.

\begin{IEEEbiography}[{\includegraphics[width=1in,height=1.25in,clip,keepaspectratio]{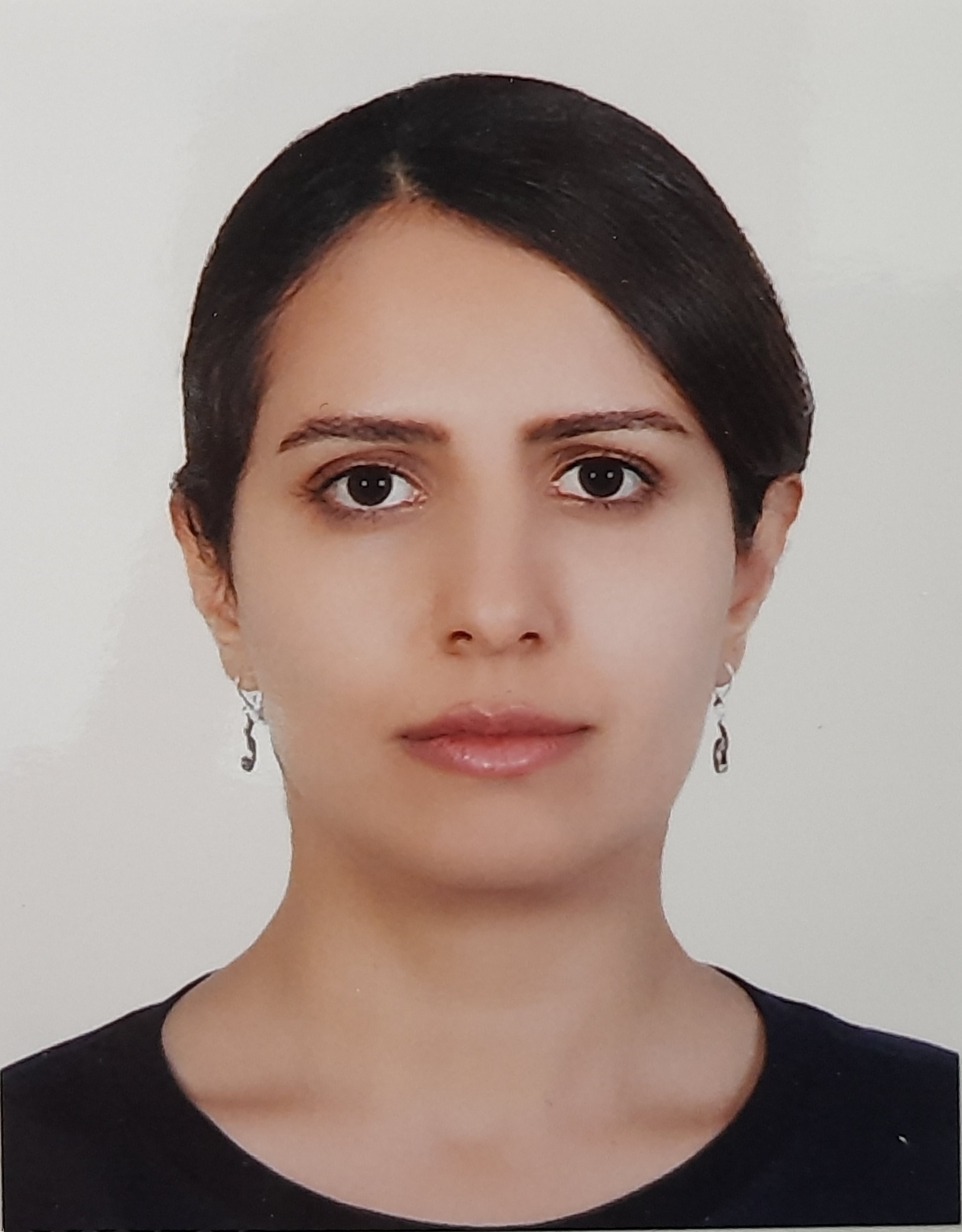}}]
{Zeinab Salehi} received the B.Sc. and M.Sc. degrees (with Distinction) in electrical engineering from Shiraz University, Shiraz, Iran, in 2015 and 2018, respectively. She is currently a Ph.D. student at the Research School of Engineering, the Australian National University, Canberra, Australia.
Her research interests include control theory and application, multi-agent systems, distributed systems, power system optimization
and planning, smart grids, renewable energy integration, machine learning, and model order reduction.
\end{IEEEbiography}


\begin{IEEEbiography}[{\includegraphics[width=1in,height=1.25in,clip,keepaspectratio]{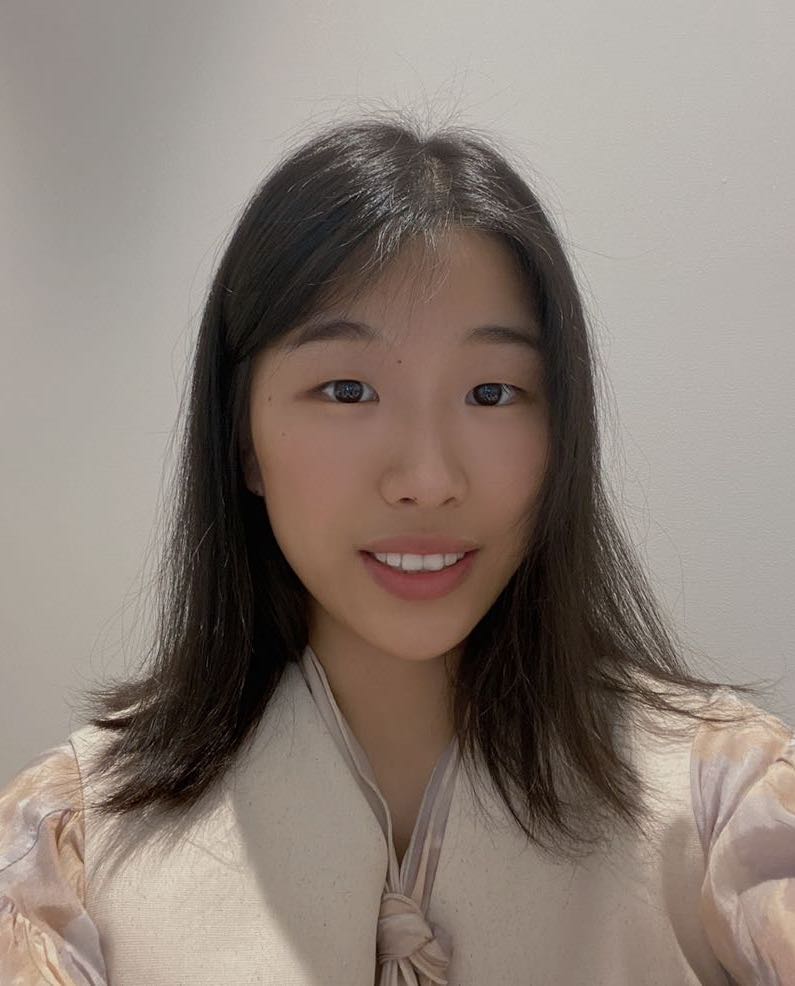}}]
{Yijun Chen} received B.Eng. Degree in Digital Media Technology from Beijing University of Posts and Telecommunications, Beijing, China in 2019. She is currently working towards a Ph.D. Degree in Control Systems Engineering with the Australian Center for Field Robotics, School of Aerospace, Mechanical and Mechatronic Engineering, University of Sydney, Sydney, Australia. Her current research interests include multi-agent systems, systems control, game theory, and decision theory.
\end{IEEEbiography}


\begin{IEEEbiography}[{\includegraphics[width=1in,height=1.25in,clip,keepaspectratio]{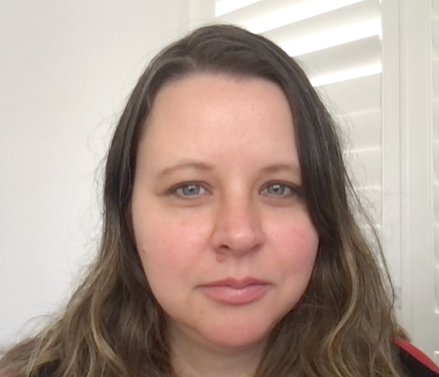}}]
{Elizabeth Ratnam} earned the BEng (Hons I) degree in Electrical Engineering in 2006, and the PhD degree in Electrical Engineering in 2016, from the University of Newcastle, Australia. She subsequently held postdoctoral research positions with the Center for Energy Research at the University of California San Diego, and at the University of California Berkeley in the California Institute for Energy and Environment (CIEE). During 2001–2012 she held various positions at Ausgrid, a utility that operates one of the largest electricity distribution networks in Australia. Dr Ratnam currently holds a Future Engineering Research Leader (FERL) Fellowship at the ANU and is a Senior Lecturer in the ANU School of Engineering. She is a Senior Member of IEEE and a Fellow of Engineers Australia. 
\end{IEEEbiography}


\begin{IEEEbiography}[{\includegraphics[width=1in,height=1.25in,clip,keepaspectratio]{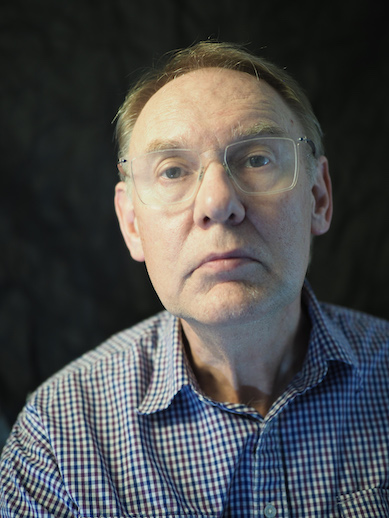}}]
{Ian R. Petersen} was born in Victoria, Australia. He received a Ph.D in Electrical Engineering in 1984 from the University of  Rochester. From 1983 to 1985 he was a Postdoctoral Fellow at the  Australian National University. From 2017 he has been a Professor at the Australian National University in the School of Engineering. He was the Interim Director of  the  School of  Engineering at the Australian National University from 2018-2019. From 1985 until 2016 he was with UNSW Canberra where he was a Scientia Professor and an Australian Research Council Laureate Fellow in the  School of Engineering and Information Technology.  He has previously been ARC Executive Director for Mathematics Information and Communications, Acting Deputy Vice-Chancellor Research for UNSW and an Australian Federation Fellow.  He has served as an Associate Editor for the IEEE Transactions on Automatic  Control, Systems and Control Letters, Automatica, IEEE Transactions on Control Systems Technology and SIAM Journal on  Control and Optimization. Currently he is an Editor for Automatica. He is a fellow of IFAC, the IEEE and the Australian Academy of Science.  His main  research interests are in robust control theory, quantum control theory and stochastic control theory. 
\end{IEEEbiography}


\begin{IEEEbiography}[{\includegraphics[width=1in,height=1.25in,clip,keepaspectratio]{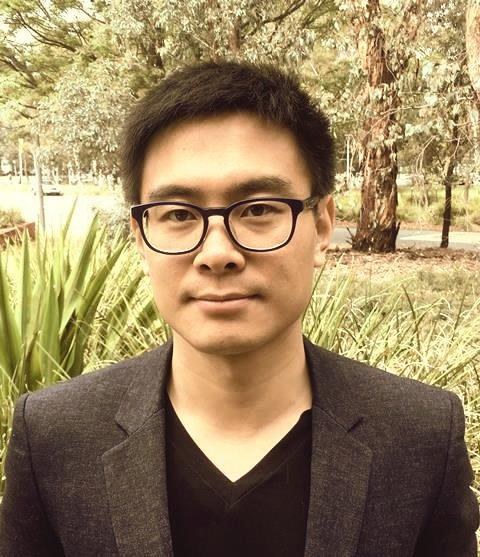}}]
{Guodong Shi} received the B.Sc. degree in mathematics and applied mathematics from the School of Mathematics, Shandong University, Jinan, China, in 2005, and the Ph.D. degree in systems theory from the Academy of Mathematics and Systems Science, Chinese Academy of Sciences, Beijing, China, in 2010. From 2010 to 2014, he was a Post-Doctoral Researcher with the ACCESS Linnaeus Centre, KTH Royal Institute of Technology, Stockholm, Sweden. From 2014 to 2018, he was with the Research School of Engineering, The Australian National University, Canberra, ACT, Australia, as a Lecturer/Senior Lecturer, and a Future Engineering Research Leadership Fellow. Since 2019, he has been with the Australian Center for Field Robotics, The University of Sydney, NSW, Australia. His research interests include distributed control systems, quantum networking and decisions, and social opinion dynamics.
\end{IEEEbiography}

\end{document}